\documentclass[a4paper,USenglish]{lipics-v2016}

\usepackage{wrapfig}
\usepackage{amsopn}
\usepackage{amsthm}
\usepackage{complexity}
\usepackage{tikz}
\usepackage{environ}
\usepackage{comment}
\usepackage{hyperref}
\usetikzlibrary{automata, calc}
\usetikzlibrary{shapes, fit}
\usepackage{xcolor,colortbl}

\DeclareFontFamily{U}{wncy}{}
\DeclareFontShape{U}{wncy}{m}{n}{<->wncyr10}{}
\DeclareSymbolFont{mcy}{U}{wncy}{m}{n}
\DeclareMathSymbol{\SShuffle}{\mathord}{mcy}{"58}
\newcommand {\Shuffle}[2]{\mathop{\SShuffle}^{#2}_{#1}}

\theoremstyle{plain}
\newtheorem{proposition}[theorem]{Proposition}

\newcommand {\set}[1]{ \{ #1 \} }
\newcommand {\Set}[2]{ \{ #1 \mid #2 \} }
\newcommand {\Nat}{\mathbb{N}}
\newcommand {\Z}{\mathbb{Z}}
\newcommand {\Powerset}{\mathcal{P}}

\newcommand {\Move}[1]{ \xrightarrow{#1} }
\DeclareMathOperator {\Context}{Context}
\DeclareMathOperator {\Int}{Int}

\newcommand {\sync}[5]{ L(M(#1,#2)) \cap L(A_{#3}(#4,#5))}
\newcommand {\APi}[2]{\Int_{#1}(#2)}
\newcommand {\pchoice}[1]{p^{ch}_{#1}}
\newcommand {\plock}[1]{p^{lock}_{#1}}
\newcommand {\qchoice}{q^{ch}}
\newcommand {\ifmap}[1]{\lambda_{#1}}
\newcommand {\execTree}{\mathcal{T}}
\newcommand {\nodeLab}{\lambda}

\newcommand {\calO}{\mathcal{O}}
\newcommand {\bigO}[1]{ \calO ( #1 ) }
\newcommand {\bigOS}[1]{ \calO^* \! ( #1 ) }
\newcommand{\smallo} [1] {o(#1)}
\newcommand {\fptred}{\leq^{fpt}}
\newcommand {\ETH}{\ComplexityFont{ETH}}
\newcommand {\SETH}{\ComplexityFont{SETH}}

\newcommand {\SM}{\ComplexityFont{Shuff}}
\newcommand {\STMA}{\ComplexityFont{STMA}}
\newcommand {\BDFAI}{\ComplexityFont{BDFAI}}
\newcommand {\kSAT}[1]{#1 \text{-} \SAT}
\newcommand {\SetCov}{\ComplexityFont{Set~Cover}}
\newcommand {\Steiner}{\ComplexityFont{Steiner~Tree}}
\newcommand {\DomSet}{\ComplexityFont{Dominating~Set}}
\newcommand {\kClique}{\ComplexityFont{k \text{-} Clique}}
\newcommand {\kkClique}{\ComplexityFont{k \times k~Clique}}

\newcommand {\BCS}{\ComplexityFont{BCS}}
\newcommand {\CS}{\ComplexityFont{CS}}
\newcommand {\cs}{\mathit{cs}}
\newcommand {\m}{\mathit{m}}
\newcommand {\SMS}[5]{ (#1, #2, (#3_{#4})_{#4 \in [1..#5]}) }
\newcommand {\Formulas}{\mathcal{F}}

\newcommand {\ilang}{\mathit{IF}}
\newcommand {\ilangof}[1]{\ilang(#1)}
\newcommand {\iilang}{\mathit{IIF}}
\newcommand {\iilangof}[1]{\iilang(#1)}
\newcommand{\threadid}{\mathit{id}}

\newcommand {\kernel}[1]{\mathcal{#1}}
\newcommand {\equivR}{\mathcal{R}}

\newcommand {\bcslsd}{\ComplexityFont{BCS}\textsf{-}\ComplexityFont{L}}
\newcommand {\bcslany}{\ComplexityFont{BCS}\textsf{-}\ComplexityFont{L}\textsf{-}\ComplexityFont{ANY}}
\newcommand {\bcslrr}{\ComplexityFont{BCS}\textsf{-}\ComplexityFont{L}\textsf{-}\ComplexityFont{RR}}
\newcommand {\bcslfix}{\ComplexityFont{BCS}\textsf{-}\ComplexityFont{L}\textsf{-}\ComplexityFont{FIX}}
\newcommand {\closure}{\downarrow}
\newcommand {\mergeop}{\otimes}
\newcommand {\omerge}[2]{\odot_{(#1,#2)}}
\newcommand {\mergeoppar}[1]{\mergeop^{#1}}

\newcommand {\contract}[2]{[#1\mapsto #2]}
\newcommand {\mgraph}{\mathit{G}}
\newcommand {\csgraphof}[1]{\mgraph(#1)}
\newcommand {\contractionprocess}{\pi}
\newcommand {\indegreeof}[1]{\mathit{indeg}(#1)}
\newcommand {\indegreeofgraph}[2]{\mathit{indeg}_{#1}(#2)}
\newcommand {\outdegreeof}[1]{\mathit{outdeg}(#1)}
\newcommand {\outdegreeofgraph}[2]{\mathit{outdeg}_{#1}(#2)}
\newcommand {\degreeof}[1]{\mathit{deg}(#1)}
\newcommand {\degreeofgraph}[2]{\mathit{deg}_{#1}(#2)}
\newcommand {\sdim}{\mathit{sdim}}
\newcommand {\sdimof}[1]{\mathit{sdim}(#1)}
\newcommand {\sdimlang}{\mathit{SDL}(\Sigma, t, \sdim)}

\newcommand {\SGI}{\ComplexityFont{SGI}}

\newcommand {\abs}[1]{\left| #1 \right|}
\newcommand {\projectto}[2]{#1\downarrow#2}
\newcommand {\maximum}{\mathit{max}}
\newcommand {\minimum}{\mathit{min}}

\DeclareMathOperator {\cw}{cw}
\newcommand {\cwof}[1]{\cw(#1)}
\DeclareMathOperator {\wi}{width}
\newcommand {\widthof}[1]{\wi(#1)}
\newcommand {\carvingdecomp}{(T,\varphi)}
\DeclareMathOperator {\Leaf}{Leaf}
\newcommand {\Leafof}[1]{\Leaf(#1)}

\DeclareMathOperator {\Pos}{Pos}

\newcommand {\TM}{\mathcal{TM}}
\newcommand {\qinit}{q_{init}}
\newcommand {\qrej}{q_{rej}}
\newcommand {\qacc}{q_{acc}}
\newcommand {\fS}{\textbf{S}}
\newcommand {\fX}{\textbf{X}}
\newcommand {\switch}{\textsc{switch}}
\newcommand {\work}{\textsc{work}}
\newcommand {\accept}{\textsc{accept}}

\makeatletter
\newcommand {\problemtitle}[1]{\gdef\@problemtitle{#1}}
\newcommand {\problemshort}[1]{\gdef\@problemshort{#1}}
\newcommand {\probleminput}[1]{\gdef\@probleminput{#1}}
\newcommand {\problemparameter}[1]{\gdef\@problemparameter{#1}}
\newcommand {\problemquestion}[1]{\gdef\@problemquestion{#1}}

\NewEnviron{paraproblem}{
  \problemtitle{}
  \problemshort{}
  \probleminput{}
  \problemparameter{}
  \problemquestion{}
  
  \BODY
  \par\addvspace{.5\baselineskip}
  \noindent
  \framebox[\textwidth]{
  \begin{tabularx}{\textwidth}{@{\hspace{\parindent}} l X c}
    \multicolumn{2}{@{\hspace{\parindent}}l}{ \normalsize \emph{\@problemtitle} \@problemshort} \\
    \normalsize \textbf{Input:} & \normalsize \@probleminput \\
    \normalsize \textbf{Parameter:} & \normalsize \@problemparameter \\ 
    \normalsize \textbf{Question:} & \normalsize \@problemquestion
  \end{tabularx}
	}
  \par\addvspace{.5\baselineskip}
}

\NewEnviron{problem}{
  \problemtitle{}\problemshort{}\probleminput{}\problemquestion{}
  \BODY
  \par\addvspace{.5\baselineskip}
  \noindent
  \framebox[\textwidth]{
  \begin{tabularx}{\textwidth}{@{\hspace{\parindent}} l X c}
    \multicolumn{2}{@{\hspace{\parindent}}l}{ \normalsize \emph{\@problemtitle} \@problemshort} \\
    \normalsize \textbf{Input:} & \normalsize \@probleminput \\
    \normalsize \textbf{Question:} & \normalsize \@problemquestion
  \end{tabularx}
	}
  \par\addvspace{.5\baselineskip}
}

\begin{document}

\title{On the Complexity of Bounded Context Switching}

\author[1]{Peter Chini}
\author[1]{Jonathan Kolberg}
\author[2]{Andreas Krebs}
\author[1]{Roland Meyer}
\author[1]{Prakash Saivasan}

\affil[1]{
	TU Braunschweig, 
	\texttt{\{p.chini, j.kolberg, roland.meyer, p.saivasan\}@tu-bs.de}
}
\affil[2]{
	Universit\"at T\"ubingen,
	\texttt{krebs@informatik.uni-tuebingen.de}
}

\authorrunning{P. Chini, J. Kolberg, A. Krebs, R. Meyer and P. Saivasan}
\maketitle

\vspace{-0.5cm}
\begin{abstract}
	Bounded context switching ($\BCS$) is an under-approximate method for finding violations to safety properties in shared memory concurrent programs.  
	Technically, $\BCS$ is a reachability problem that is known to be $\NP$-complete.
	Our contribution is a parameterized analysis of $\BCS$.

	The first result is an algorithm that solves $\BCS$ when parameterized by the number of context switches ($\cs$) and the size of the memory ($m$) in $\bigOS{m^{cs} \cdot 2^{cs}}$.
	This is achieved by creating instances of the easier problem $\SM$ which we solve via fast subset convolution. 
	We also present a lower bound for $\BCS$ of the form $\m^{o(\cs / \log(\cs))}$, based on the exponential time hypothesis.
	Interestingly, closing the gap means settling a conjecture that has been open since FOCS'07. 
	Further, we prove that $\BCS$ admits no polynomial kernel.
	
	Next, we introduce a measure, called scheduling dimension, that captures the complexity of schedules. 
	We study $\BCS$ parameterized by the scheduling dimension ($\sdim$) and show that it can be solved in $\bigOS{(2m)^{4\sdim}4^{t}}$, where $t$ is the number of threads. 
	We consider variants of the problem for which we obtain (matching) upper and lower bounds.
\end{abstract}
\vspace{-1cm}
\section{Introduction}
\label{Section:Introduction}

Concurrent programs where several threads interact through a shared memory can be found essentially everywhere where performance matters, in particular in critical infrastructure like operating systems and libraries. 
The asynchronous nature of the communication makes these programs prone to programming errors. 
As a result, substantial effort has been devoted to developing automatic verification tools. 
The current trend for shared memory is bug-hunting: Algorithms that look for misbehavior in an under-approximation of the computations.

The most prominent method in the under-approximate verification of shared-memory concurrent programs is bounded context switching~\cite{QadeerR05}.
A context switch occurs when one thread leaves the processor for another thread to be scheduled.
The idea of bounded context switching ($\BCS$) is to limit the number of times the threads may switch the processor.
Effectively this limits the communication that can occur between the threads.
(Note that there is no bound on the running time of each thread.)
Bounded context switching has received considerable attention~\cite{TorreMP07,AtigBTA08,AtigBT08,Atig10,TorreMP10,TorreN11,AtigBKS14,LeonFHM17} for at least two reasons.
First, the under-approximation has been demonstrated to be useful in numerous experiments, in the sense that synchronization bugs show up in few context switches~\cite{MusuvathiQ07}. 
Second, compared to ordinary algorithmic verification, $\BCS$ is algorithmically appealing, with the complexity dropping from $\PSPACE$ to $\NP$ in the case of Boolean programs. 

The hardness of verification problems, also the $\NP$-hardness of $\BCS$, is in sharp contrast to the success that verification tools see on industrial instances. 
This discrepancy between the worst-case behavior and efficiency in practice has also been observed in other areas within algorithmics.  
The response was a line of research that refines the classical worst-case complexity.
Rather than only considering problems where the instance-size determines the running time, so-called parameterized problems identify further parameters that give information about the structure of the input or the shape of solutions of interest.
The complexity class of interest consists of the so-called fixed-parameter tractable problems. 
A problem is fixed-parameter tractable if the parameter that has been identified is indeed responsible for the non-polynomial running time or, phrased differently, the running time is $f(k)p(n)$ where $k$ is the parameter, $n$ is the size of the input, $f$ is a computable function and $p$ is a polynomial.

Within fixed-parameter tractability, the recent trend is a fine-grained analysis to understand the precise functions $f$ that are needed to solve a problem. 
From an algorithmic point of view, an exponential dependence on $k$, at best linear so that $f(k) = 2^k$, is particularly attractive. 
There are, however, problems where algorithms running in $2^{\smallo{k\log( k )}}$ are unlikely to exist.
As common in algorithmics, unconditional lower bounds are hard to achieve, and none are known that separate $2^k$ and $2^{k\log (k)}$.
Instead, one works with the so-called exponential time hypothesis ($\ETH$): After decades of attempts, $n$-variable $\kSAT{3}$ is not believed to admit an algorithm of running time $2^{\smallo{n}}$. 
To derive a lower bound for a problem, one now shows a reduction from $n$-variable $\kSAT{3}$ to the problem such that a running time in $2^{\smallo{k\log (k)}}$ means $\ETH$ breaks.

The contribution of our work is a fine-grained complexity analysis of the bounded context switching under-approximation. 
We propose algorithms as well as matching lower bounds in the spectrum $2^k$ to $k^k$.
This work is not merely motivated by explaining why verification works in practice.
Verification tasks have also been shown to be hard to parallelize.
Due to the memory demand, the current trend in parallel verification is lock-free data structures \cite{BarnatBCMRS06}.
So far, GPUs have not seen much attention.
With an algorithm of running time $2^kp(n)$, and for moderate $k$, say $12$, one could run in parallel 4096 threads each solving a problem of polynomial effort.

When parameterized only by the context switches, $\BCS$ is quickly seen to be $\W[1]$-hard and hence does not admit an $\FPT$-algorithm.
Since it is often the case that shared-memory communication is via signaling (flags), memory requirements are not high.
We additionally parameterize by the memory.
Our study can be divided into two parts.

We first give a parameterization of $\BCS$ (in the context switches and the size of the memory) that is \emph{global} in the sense that all threads share the budget of $\cs$ many context switches.
For the upper bound, we show that the problem can be solved in $\bigOS{m^{\cs}2^{\cs}}$. 
We first enumerate the sequences of memory states at which the threads could switch context, and there are $m^{\cs}$ such sequences where $m$ is the size of the memory.
For a given such sequence, we check a problem called $\SM$: Given a memory sequence, do the threads have computations that justify  the sequence (and lead to their accepting state).
Here, we use fast subset convolution to solve $\SM$ in $\bigOS{2^{\cs}}$.
Note that $\SM$ is a problem that may be interesting in its own right.
It is an under-approximation that still leaves much freedom for the local computations of the threads.
Indeed, related ideas have been used in testing~\cite{GibbonsKorach1997,Cantin2005,FaMa09,Furbach2015}.

For the lower bound, the finding is that the global parameterization of $\BCS$ is closely related to subgraph isomorphism.
Whereas the reduction is not surprising, the relationship is, with $\SGI$ being one of the problems whose fine-grained complexity is not fully understood. 
Subgraph isomorphism can be solved in $\bigOS{n^{k}}$ where $k$ is the number of edges in the graph that is to be embedded.
The only lower bound, however, is $n^{o(k/\log k)}$, and has, to the best of our knowledge, not been improved since FOCS'07 \cite{Marx07,Marx2010}.
However, the believe is that the $\log k$-gap in the exponent can be closed.
We show how to reduce $\SGI$ to the global version of $\BCS$, and obtain a $m^{o(\cs/\log \cs)}$ lower bound.
Phrased differently, $\BCS$ is harder than $\SGI$ but admits the same upper bound.
So once Marx' conjecture is proven, we obtain a matching bound.
If we proved a lower upper bound, we had disproven Marx' conjecture.

Our second contribution is a study of $\BCS$ where the parameterization is \emph{local} in the sense that every thread is given a budget of context switches.
Here, our focus is on the scheduling.
We associate with computations so-called scheduling graphs that show how the threads take turns. 
We define the scheduling dimension, a measure on scheduling graphs (shown to be closely related to carving width) that captures the complexity of a schedule.
Our main finding is a fixed-point algorithm that solves the local variant of $\BCS$ exponential only in the scheduling dimension and the number of threads.
We study variants where only the budget of context switches is given, the graph is given, and where we assume round robin as a schedule.
Verification under round robin has received quite some attention \cite{BouajjaniEP11, MusuvathiQ07, LalR08}. 
Here, we show that we get rid of the exponential dependence on the number of threads and obtain an $\bigOS{m^{4\cs}}$ upper bound.
We complement this by a matching lower bound.

The following table summarizes our results and highlights the main findings in gray.
\begin{wraptable}{r}{6.2cm}
    \centering
    \begin{tabular}{ |m{1.25cm} | m{1.8cm}| m{1.85cm}| }
        \hline
        \scriptsize{Problem} & \scriptsize{Upper Bound} & \scriptsize{Lower Bound} \\ 
        \hline
        \hline
        \scriptsize{\SM}  & \cellcolor{black!15}\scriptsize{$\bigOS{2^{k}}$} & \scriptsize{$(2 - \varepsilon)^{k}$} \\ 
        \hline
        \scriptsize{\BCS}  & \cellcolor{black!15}\scriptsize{$\bigOS{m^{\cs}2^{\cs}}$} &\cellcolor{black!15}\scriptsize{$m^{o(\cs/\log \cs)}$, no poly. kernel} \\ 
        \hline
        \scriptsize{\bcslrr}   & \scriptsize{$\bigOS{m^{4\cs}}$} & \scriptsize{$2^{o(\cs\log(m))}$} \\
        \hline        
        \scriptsize{\bcslfix}   & \cellcolor{black!15} \scriptsize{$\bigOS{(2m)^{4\sdim}}$} & \scriptsize{$2^{o(\sdim\log(m))}$} \\
        \hline
        \scriptsize{\bcslsd}  & \cellcolor{black!15} \scriptsize{$\bigOS{(2m)^{4\sdim}4^t}$} & \scriptsize{$2^{o(\sdim\log(m))}$} \\
        \hline
    \end{tabular}
\end{wraptable}
The organization is by expressiveness, measured in terms of the amount of computations that an analysis explores. 
Considering shuffle membership \SM\ as an under-approximate analysis in its own right,
\SM\ is less expressive than the globally parameterized \BCS. 
\BCS\ is less expressive than round robin \bcslrr, which is a special instance of fixing the scheduling graph \bcslfix.
The most liberal parameterization is via the scheduling dimension \bcslsd.
In the paper, we present algorithms for the case where threads are finite state. 
Our results also hold for more general classes of programs, notably recursive ones.
The only condition that we require is that the chosen automaton model for the threads has a polynomial time decision procedure for checking non-emptiness when intersected with a regular language.

There have been previous efforts in studying fixed-parameter tractable algorithms for automata and verification-related problems.
In \cite{EneaF16} , the authors introduced the notion of conflict serializability under TSO and gave an $\FPT$-algorithm for checking serializability.
In \cite{FarzanM09}, the authors studied the complexity of predicting atomicity violation on concurrent systems and showed that no $\FPT$ solution is possible for the same.
In \cite{DemriLS02}, various model checking problems for synchronized executions on parallel components were considered and proven to be intractable.
Parameterized complexity analyses for two problems on automata were given in \cite{Fernau2015}.
Also in \cite{Wareham2000}, a complete parameterized complexity analysis of the intersection non-emptiness problem was shown.

Verification of concurrent systems has received considerable attention. The parameterized verification of concurrent systems was studied in \mbox{\cite{Durand-Gasselin15, EsparzaGM13, FortinMW16, Hague11, TorreMP10}}.
Concurrent shared memory system with fixed number of threads were also studied in \cite{AtigBKS14,AtigBT08,AtigBQ09}.

\section{Preliminaries}
\label{Section:Prelim}
We define the bounded context switching problem~\cite{QadeerR05} of interest and recall the basics on fixed-parameter tractability following~\cite{Downey2013, Flum2006}.
\vspace{-0.4cm}
\subparagraph*{Bounded Context Switching.}
We study the safety verification problem for shared memory concurrent programs. 
To obtain precise complexity results, it is common to assume both the number of threads and the data domain to be finite. 
Safety properties partition the states of a program into \emph{unsafe} and \emph{safe} states.
Hence, checking safety amounts to checking whether no unsafe state is reachable.
In the following, we develop a language-theoretic formulation of the reachability problem that will form the basis of our study.

We model the shared memory as a (non-deterministic) finite automaton of the form $M=(Q, \Sigma, \delta_M, q_0, q_f)$. 
The states $Q$ correspond to the data domain, the set of values that the memory can be in.
The initial state $q_0\in Q$ is the value that the computation starts from.  
The final state $q_f\in Q$ reflects the reachability problem.
The alphabet $\Sigma$ models the set of operations. 
Operations have the effect of changing the memory valuation, formalized by the transition relation $\delta_M\ \subseteq Q\times \Sigma\times Q$. 
We generalize the transition relation to words $u\in \Sigma^*$. 
The set of valid sequences of operations that lead from a state $q$ to another state $q'$ is the language \mbox{$L(M(q, q')):=\Set{u\in \Sigma^*}{q'\in \delta_M(q, u)}$}. 
The language of $M$ is $L(M):=L(M(q_0, q_f))$. 
The size of $M$, denoted $\abs{M}$,  is the number of states. 

We also model the threads operating on the shared memory $M$ as finite automata \mbox{$A_{\threadid}=(P, \Sigma\times\set{\threadid}, \delta_A, p_0, p_f)$}.  
Note that they use the alphabet $\Sigma$ of the shared memory, indexed by the name of the thread. 
The index will play a role when we define the notion of context switches below. 
The automaton $A_{\threadid}$ is nothing but the control flow graph of the thread $\threadid$. 
Its language is the set of sequences of operations that the thread may potentially execute to reach the final state.  
As the thread language does not take into account the effect of the operations on the shared memory, not all these sequences will be feasible. 
Indeed, the thread may issue a command $\mathit{write}(x, 1)$ followed by $\mathit{read}(x, 0)$, which the automaton for the shared memory will reject.
The computations of $A$ that are actually feasible on the shared memory are given by the intersection $L(M)\cap L(A_{\threadid})$.
Here, we silently assume the intersection to project away the second component of the thread alphabet.

A concurrent program consists of multiple threads $A_1$ to $A_t$ that mutually influence each other by accessing the same memory $M$. 
We mimic this influence by interleaving the thread languages, formalized with the shuffle operator $\Shuffle{}{}$. 
Consider languages $L_1\subseteq \Sigma_1^*$ and $L_2\subseteq\Sigma_2^*$ over disjoint alphabets $\Sigma_1\cap\Sigma_2=\emptyset$.
The shuffle of the languages contains all words over the union of the alphabets where the corresponding projections ($\projectto{-}{-}$) belong to the operand languages, $L_1\Shuffle{}{}L_2:=\Set{u\in (\Sigma_1\cup \Sigma_2)^*}{\projectto{u}{\Sigma_i}\in L_i \cup \set{\varepsilon}, i=1,2}$.

With these definitions in place, a \emph{shared memory concurrent program (SMCP)} is a tuple $S=\SMS{\Sigma}{M}{A}{i}{t}$. 
Its language is $L(S)\ :=\ L(M)\ \cap\ (\ \Shuffle{i \in [1..t]}{} L(A_i)\ )$. The safety verification problem induced by the program is to decide whether $L(S)$ is non-empty. 

We formalize the notion of context switching. 
Every word in the shuffle of the thread languages, $u\in \Shuffle{i \in [1..t]}{} L(A_i)$, has a unique decomposition into maximal infixes that are generated by the same thread. 
Formally, $u = u_1\ldots u_{\cs+1}$ so that there is a function \mbox{$\varphi:[1..\cs+1]\rightarrow [1..t]$} satisfying $u_i\in (\Sigma\times \set{\varphi(i)})^+$ and $\varphi(i)\neq \varphi(i+1)$ for all $i\in [1..\cs]$. 
We refer to the $u_i$ as contexts and to 
the thread changes between $u_i$ to $u_{i+1}$ as \emph{context switches}.
So $u$ has $\cs+1$ contexts and $\cs$ context switches.
Let $\Context(\Sigma, t, \cs)$ denote the set of all words (over $\Sigma$
with $t$ threads) that have at most $\cs$-many context switches.
The \emph{bounded context switching} under-approximation limits the safety
verification task to this language.
\begin{problem}
	\problemtitle{Bounded Context Switching}
	\problemshort{($\BCS$)}
	\probleminput{ An SMCP $S = \SMS{\Sigma}{M}{A}{i}{t}$ and a bound $\cs\in \Nat$. }
	\problemquestion{Is $L(S) \cap \Context(\Sigma, t, \cs)\neq \emptyset$ ?}
\end{problem}
\vspace{-0.3cm}
\subparagraph*{Fixed Parameter Tractability.} 
$\BCS$ is $\NP$-complete by~\cite{Esparza2014}, even for unary alphabets. 
Our goal is to understand which instances can be solved efficiently and, in turn, what makes an instance hard. 
Parameterized complexity addresses these questions.

A \emph{parameterized problem} $L$ is a subset of $\Sigma^*\times\mathbb N$. 
The problem is \emph{fixed-parameter tractable (FPT)} if there is a deterministic algorithm that, given $(x,k)\in\Sigma^*\times\mathbb N$, decides $(x,k)\in L$ in time $f(k)\cdot \abs{x}^{O(1)}$. 
Here, $f$ is a computable function that only depends on the parameter~$k$.
It is common to denote the runtime by $\bigOS{f(k)}$ and suppress the polynomial part.

While many parameterizations of $\NP$-hard problems were proven to be fixed-parameter tractable, there are problems that are unlikely to be $\FPT$. 
A famous example that we shall use is $\kClique$, the problem of finding a clique of size $k$ in a given graph. 
$\kClique$ is complete for the complexity class $\W[1]$, and $\W[1]$ hard problems are believed to lie outside $\FPT$. 

A theory of relative hardness needs an appropriate notion of reduction. 
Given parameterized problems $L, L'\subseteq\Sigma^*\times\mathbb N$, we say that $L$ is \emph{reducible} to $L'$ via a \emph{parameterized reduction}, denoted by $L \leq^{\mathit{fpt}} L'$, if there is an algorithm that transforms an input $(x,k)$ to an input $(x',k')$ in time $g(k) \cdot n^{O(1)}$ so that $(x,k) \in L$ if and only if $(x',k') \in L'$.
Here, $g$ is a computable function and $k'$ is computed by a function only dependent on $k$.

For $\BCS$, a first result is that a parameterization by the number of context switches and additionally by the number of threads, denoted by $\BCS(\cs, t)$, is not sufficient for $\FPT$: The problem is $\W[1]$-hard.
It remains in $\W[1]$ if we only parameterize by the context switches.  
\begin{proposition}\label{Theorem:BCSIntractable}
	$\BCS(\cs)$ and $\BCS(\cs,t)$ are both $\W[1]$-complete.
\end{proposition}

The runtime of an $\FPT$-algorithm is dominated by $f$.
The goal of \emph{fine-grained complexity theory} is to give upper and lower bounds on this non-polynomial function.  
For lower bounds, the problem that turned out to be hard is $n$-variable $\kSAT{3}$.
The \emph{Exponential Time Hypothesis} ($\ETH$) is that the problem does not admit a $2^{o(n)}$-time algorithm~\cite{Impagliazzo2001}. 
We will prove a number of lower bounds that hold, provided $\ETH$ is true.

In the remainder of the paper, we consider parameterizations of $\BCS$ that are $\FPT$. 
Our contribution is a fine-grained complexity analysis.

\section{Global Parametrization}
\label{Section:BCS}
Besides the number of context switches $\cs$, we now consider the size $m$ of the memory as a parameter of $\BCS$.  
This parameterization is practically relevant and, as we will show, algorithmically appealing. 
Concerning the relevance, note that communication over the shared memory is often implemented in terms of flags. 
Hence, when limiting the size of the memory we still explore a large part of the computations. 
\vspace{-0.3cm}
\subparagraph*{Upper Bounds.}
\label{Subsection:BCSUpper}
The idea of our algorithm is to decompose $\BCS$ into exponentially many
instances of the easier problem shuffle membership ($\SM$) defined below.
Then we solve $\SM$ with fast subset convolution. 
To state the result, let the given instance of $\BCS$ be $S=\SMS{\Sigma}{M}{A}{i}{t}$ with bound $\cs$. 
To each automaton $A_i$, our algorithm will associate another automaton $B_i$ of size polynomial in $A_i$.
Let  $b =\max_{i \in [1..t]} \abs{B_i}$.
Moreover, let $\SM(b, k,t ) = \bigO{2^k \cdot t \cdot k \cdot ( b^2 + k \cdot \mathit{bc}(k) )}$ be the complexity of solving the shuffle problem. 
The factor $\mathit{bc}(k)$ appears as we need to multiply $k$-bit integers (see below). 
The currently best known running time is $\mathit{bc}(k) = k \log k \cdot 2 ^{\bigO{\log^* \!\! k}}$~\cite{Furer2009, Harvey2016}. 
\begin{theorem}\label{Theorem:BCSupper}
	$\BCS$ can be solved in $\bigO{ m^{cs+1} \cdot \SM(b, \cs+1,t )  + t \cdot m^3 \cdot b^3}$. 
\end{theorem}

We decompose \BCS\ along \emph{interface sequences}. 
Such an interface sequence is a word \mbox{$\sigma=(q_1, q_1')\ldots (q_{k}, q_{k}')$} over pairs of states of the memory automaton $M$.
The length is $k$.
An interface sequence is \emph{valid} if $q_1$ is the initial state of the memory automaton, $q_k'$ the final state, and $q_i'=q_{i+1}$ for $i\in [1..k-1]$. 
Consider a word $u \in L(S)$ with contexts $u = u_1 \dots u_{m}$.
An interface sequence $\sigma=(q_0, q_1)(q_1, q_2) \dots (q_{m-1}, q_{m})$ is \emph{induced} by $u$, if there is an accepting run of $M$ on $u$ such that for all $i \in [1..m]$, $q_i$ is the state reached by  $M$ upon reading $u_1 \dots u_i$. 
Note that we only consider the states that occur upon context switches.
Moreover, induced sequences are valid by definition.
Finally, note that a word with $\cs$-many context switches induces an interface sequence of length precisely $\cs+1$.
We define $\iilangof{S}\subseteq (Q\times Q)^*$ to be the \emph{language of all induced interface sequences}.

Induced interface sequences witness non-emptiness of $L(S)$:
$L(S) \neq \emptyset$ iff $\iilangof{S} \neq \emptyset$.  
Since the number of context switches is bounded by $\cs$, we can thus iterate over all sequences in $(Q\times Q)^{\leq \cs+1}$ and test each of them for being an induced interface sequence, i.e. an element of $\iilangof{S}$.
Since induced sequences are valid, there are at most $m^{\cs+1}$ sequences to test.

Before turning to this test, we do a preprocessing step that removes the dependence on the memory automaton $M$.
To this end, we define the \emph{interface language} $\ilangof{A_{\threadid}}$ of a thread.
It makes visible the state changes on the shared memory that the contexts of this thread may induce. 
Formally, the interface language consists of all interface sequences $(q_1, q_1')\ldots (q_{k}, q_{k}')$ so that $L(A_{\threadid})\cap (\ L(M(q_1, q_1'))\ldots L(M(q_{k}, q_{k}'))\ )\neq \emptyset$. 
These sequences do not have to be valid as the thread may be interrupted by others. 
Below, we rely on the fact that $\ilangof{A_{\threadid}}$ is again a regular language, a representation of which is easy to compute.
\begin{lemma}\label{Lemma:MemorySequenceEquivalence}
(i)	We have $\iilangof{S}=\SShuffle_{i \in [1..t]} \ilangof{A_i}\cap \Set{\sigma\in (Q\times Q)^*}{\sigma\text{ valid}}$. (ii) One can compute in time 
$\bigO{\abs{A_{\threadid}}^3 \cdot \abs{M}^3}$ an automaton $B_{\threadid}$ with $L(B_{\threadid}) = \ilangof{A_{\threadid}}$.
\end{lemma}

With the above reasoning, and since the analysis is restricted to $\cs$-many context switches, the task is to check whether a valid sequence $\sigma\in(Q\times Q)^{\cs+1}$ is included in the shuffle $\SShuffle_{i \in [1..t]} L(B_i)$. 
This means we address the following problem:

\begin{problem}
	\problemtitle{Shuffle Membership}
	\problemshort{($\SM$)}
	\probleminput{NFAs $(B_i)_{i \in [1..t]}$ over the alphabet $\Gamma$, an integer $k$, and a word $w \in \Gamma^k$.}
	\problemquestion{Is $w$ in $\SShuffle_{i \in [1..t]} L(B_i)$ ?}
\end{problem}

We obtain the following upper bound, with $b$ and $\mathit{bc}(k)$ as defined above.

\begin{theorem}\label{Theorem:SMupper}
	$\SM$ can be solved in time $\bigO{2^k \cdot t \cdot k \cdot ( b^2 + k \cdot \mathit{bc}(k) )}$.
\end{theorem}

Our algorithm is based on \emph{fast subset convolution} \cite{Bjorklund2007}, an algebraic technique for summing up partitions of a given set.
Typically, fast subset convolution is applied to graph problems:
Björklund et al. \cite{Bjorklund2007} used it to present the first $\bigOS{2^k}$-time algorithm for the $\Steiner$ problem with $k$ terminals and bounded edge weights.
Cygan et al. incorporated a generalized version as a subprocedure in applications of their \emph{Cut} \& \emph{Count} technique~\cite{Cygan2011}.
Variants of $\DomSet$ parameterized by treewidth were solved by van Rooij et al. in \cite{Rooij2009} using fast subset convolution.
We are not aware of an automata-theoretic application.

Let $f,g: \Powerset(B) \rightarrow \Z$ be two functions from the powerset of a $k$-element set $B$ to the ring of integers. 
The \emph{convolution} of $f$ and $g$ is the function $f \ast g : \Powerset(B) \rightarrow \Z$ that maps a subset $S \subseteq B$ to the sum $\sum_{U \subseteq S} f(U) g(S \setminus U)$. 
Note that the convolution is associative.
There is a close connection to partitions.
For $t \in \Nat$, a $t$-\emph{partition} of a set $S$ is a tuple $(U_1, \dots, U_t)$ of subsets of $S$ such that $U_1 \cup \dots \cup U_t = S$ and $U_i \cap U_j = \emptyset$ for all $i \neq j$.
Now it is easy to see that the convolution of $t$ functions $f_i : \Powerset(B) \rightarrow \Z, i \in [1..t]$, sums up all $t$-partitions of $S$:
\vspace{-0.1cm}
\begin{align*}\label{Formula:FSCt}
	(f_1 \ast \cdots \ast f_t) (S) = \sum_{\substack{(U_1, \dots, U_t)\\ \text{ is a $t$-parition of }S}} f_1(U_1) \cdots f_t(U_t)\ .
\end{align*}
\vspace{-0.3cm}

To apply the convolution, we give a characterization of $\SM$ in terms of partitions.
Let $((B_i)_{i\in[1..t]}, k, w)$ be an instance of $\SM$.
The following observation is crucial.
The word $w$ lies in the shuffle of the $L(B_i)$ if and only if there are non-overlapping, possibly empty (scattered) subwords $w_1, \dots, w_t$ of $w$ that decompose $w$ and that satisfy $w_i \in L(B_i) \cup \set{\varepsilon}$ for all $i \in [1..t]$. 
By scattered, we mean that the subwords do not have to form an infix of $w$. 
Such a decomposition induces a $t$-partition $(U_1, \dots, U_t)$ of the set of positions $\Pos = \set{1, \dots, k}$ of $w$, where each $U_i$ holds exactly the positions of $w_i$. 
In turn, given a $t$-partition $(U_1, \dots, U_t)$ of $\Pos$, we can derive a decomposition of $w$ by setting $w_i = w[U_i]$ for all $i \in [1..t]$.
Here, $w[U_i]$ is the projection of $w$ to the positions in $U_i$. 
Hence, $w$ lies in the shuffle if and only if there is a $t$-partition $(U_1, \dots, U_t)$ of $\Pos$ such that $w[U_i] \in L(B_i) \cup \set{\varepsilon}$ for all $i \in [1..t]$.

To express the language membership in $L(B_i)$ in terms of functions, we employ the characteristic functions $f_i : \Powerset(\Pos) \rightarrow \Z$ that map a set $S$ to $1$ if $w[S] \in L(B_i) \cup \set{\varepsilon}$, and to $0$ otherwise.
By the above formula, it follows that $(f_1 \ast \cdots \ast f_t)(\Pos) > 0$ if and only if there is a $t$-partition $(U_1, \dots, U_t)$ of $\Pos$ such that $f_i(U_i) = 1$ for $i \in [1..t]$.
Altogether, we have proven the following lemma:
\begin{lemma}
	The word $w \in \Gamma^k$ is in $\SShuffle_{i \in [1..t]} L(B_i)$ if and only if $(f_1 \ast \cdots \ast f_t)(\Pos) > 0$.
\end{lemma}

Our algorithm for $\SM$ computes the characteristic functions $f_i$ and $t-1$ convolutions to obtain $f_1 \ast \cdots \ast f_t$.
Then it evaluates the convolution at the set $\Pos$. 
Computing and storing a value $f_i(S)$ for a subset $S \subseteq \Pos$ takes time $\bigO{ k \cdot b^2 }$ since we have to test membership of a word of length at most $k$ in $B_i$.
Hence, computing all $f_i$ takes time $\bigO{2^k \cdot t \cdot k \cdot b^2}$.
Due to Björklund et al. \cite{Bjorklund2007}, we can compute the convolution of two functions $f,g : \Powerset(\Pos) \rightarrow \Z$ in $\bigO{2^k \cdot k^2}$ operations in $\Z$.
Furthermore, if the ranges of $f$ and $g$ are bounded by $C$, we have to perform these operations on $\bigO{k \log C}$-bit integers \cite{Bjorklund2007}. 
Since the characteristic functions $f_i$ have ranges bounded by a constant, we only need to compute with $\bigO{k}$-bit integers. 
Hence, the $t-1$ convolutions can be carried out in time $\bigO{2^k \cdot k^2 \cdot (t-1) \cdot \mathit{bc}(k)}$. 
Altogether, this proves Theorem \ref{Theorem:SMupper}.

\vspace{-0.4cm}
\subparagraph*{Lower Bound for Bounded Context Switching.}
We prove a lower bound for the $\NP$-hard $\BCS$ by reducing the \emph{Subgraph Isomorphism} problem to it. 
The result is such that it also applies to $\BCS(\cs)$ and $\BCS(\cs, m)$.  
We explain why the result is non-trivial.

In fine-grained complexity, lower bounds for $\W[1]$-hard problems are often obtained by reductions from $\kClique$.  
Chen et al. \cite{Chen2006} have shown that $\kClique$ cannot be solved in time $f(k)n^{o(k)}$ for any computable function $f$, unless $\ETH$ fails. 
To transport the lower bound to a problem of interest, one has to construct a parameterized reduction that blows up the parameter only linearly. 
In the case of $\BCS$, this fails.
We face a well-known problem which was observed for reductions using edge-selection gadgets \cite{Marx2010, Cygan2015}:
A reduction from $\kClique$ would need to select a clique candidate of size $k$ and check whether every two vertices of the candidate share an edge.
This needs $\bigO{k^2}$ communications between the chosen vertices,
which translates to $\bigO{k^2}$ context switches.
Hence, we only obtain $n^{o(\sqrt{k})}$ as a lower bound.

To overcome this, we follow Marx~\cite{Marx2010} and give a reduction from \emph{Subgraph Isomorphism} ($\SGI$).
This problem takes as input two graphs $G$ and $H$ and asks whether $G$ is isomorphic to a subgraph of $H$.
This means that there is an injective map $\varphi: V(G) \rightarrow V(H)$ such that for each edge $(u,v)$ in $G$, the pair $(\varphi(u), \varphi(v))$ is an edge in $H$.
We use $V(G)$ to denote the vertices and $E(G)$ to denote the edges of a graph $G$. 
Marx has shown that $\SGI$ cannot be solved in time $f(G)n^{o(k/\log k)}$, where $k$ is the number of edges of $G$, unless $\ETH$ fails. 
In our reduction, the number of edges is mapped linearly to the number of context switches.
\begin{theorem}\label{Theorem:BCSLower}
	Assuming $\ETH$, there is no $f$ s.t. $\BCS$ can be solved in $f(\cs) n^{o(\cs / \log(\cs))}$.
\end{theorem}
Roughly, the idea is this:
The alphabet $V(G) \times V(H)$ describes how the vertices of $G$ are mapped to vertices of $H$. 
Now we can use the memory $M$ to output all possible injective maps from $V(G)$ to $V(H)$. 
There is one thread $A_i$ for each edge of $G$.
Its task is to verify that the edges of $G$ get mapped to edges of $H$.

Note that Theorem \ref{Theorem:BCSLower} implies a lower bound for the $\FPT$-problem $\BCS(\cs,\m)$. 
It cannot be solved in $\m^{o(\cs / \log(\cs))}$ time, unless $\ETH$ fails.
\vspace{-0.3cm}
\subparagraph*{Lower Bound for Shuffle Membership.}
We prove it unlikely that $\SM$ can be solved in $\bigOS{(2 - \delta)^k}$ time, for a $\delta > 0$.
Hence, the $\bigOS{2^k}$-time algorithm above may be optimal.
We base our lower bound on a reduction from $\SetCov$.
An instance consists of a family of sets $(S_i)_{i \in [1..m]}$ over an universe $U = \bigcup_{i\in [1..m]} S_i$, and an integer $t \in \Nat$.
The problem asks for $t$ sets $S_{i_1}, \dots, S_{i_t}$ from the family such that $U = \bigcup_{j \in [1..t]} S_{i_j}$.

We are interested in a parameterization of the problem by the size $n$ of the universe.
It was shown that this parameterization admits an $\bigOS{2^n}$-time algorithm~\cite{Fomin2004}.
But so far, no $\bigOS{(2 - \varepsilon)^n}$-time algorithm was found, for an $\varepsilon > 0$.
Actually, the authors of \cite{Cygan2016} conjecture that the existence of such an algorithm would contradict the \emph{Strong Exponential Time Hypothesis} ($\SETH$) \cite{Impagliazzo2001, Calabro2009}. 
This is the assumption that \mbox{$n$-variable} $\SAT$ cannot be solved in $\bigOS{(2 - \varepsilon)^n}$ time, for an $\varepsilon > 0$ ($\SETH$ implies $\ETH$).
By now, there is a list of lower bounds based on $\SetCov$~\cite{Bjorklund2016,Cygan2016}. 
We add $\SM$ to this list.
\begin{proposition}\label{Theorem:LowerBoundSM}
	If $\SM$ can be solved in $\bigOS{(2-\delta)^k}$ time for a $\delta > 0$,  then $\SetCov$ can be solved in $\bigOS{(2-\varepsilon)^n}$ time for an $\varepsilon > 0$.
\end{proposition}
\vspace{-0.5cm}
\subparagraph*{Lower Bound on the Size of the Kernel.}
Kernelization is a preprocessing technique for parameterized problems that transforms a given instance to an equivalent instance of size bounded by a function in the parameter.
It is well-known that any $\FPT$-problem admits a kernelization and any kernelization yields an $\FPT$-algorithm \cite{Cygan2015}.
The search for small problem-kernels is ongoing research.
A survey can be found in \cite{Lokshtanov2012}.

There is also the opposite approach, disproving the existence of a kernel of polynomial size~\cite{Bodlaender2009, Fortnow2011}. 
Such a result indicates hardness of the problem at hand, and hence serves as a lower bound. 
Technically, the existence of a polynomial kernel is linked to the inclusion \mbox{$\NP \subseteq \coNP/\poly$}. 
The latter is unlikely as it would cause a collapse of the polynomial hierarchy to the third level \cite{Yap1983}. 
Based on this approach, we show that $\BCS(\cs,\m)$ does not admit a kernel of polynomial size.
We introduce the needed notions, following~\cite{Cygan2015}.

A \emph{kernelization} for a parameterized problem $Q$ is an algorithm that, given an instance $(I,k)$, returns an equivalent instance $(I',k')$ in polynomial time such that \mbox{$|I'| + k' \leq g(k)$} for some computable function $g$. 
If $g$ is a polynomial, $Q$ is said to admit a \emph{polynomial kernel}.

We also need \emph{polynomial equivalence relations}.
These are equivalence relations on $\Sigma^*$, with $\Sigma$ some alphabet, such that: 
(1) There is an algorithm that, given $x,y \in \Sigma^*$, decides whether $(x,y) \in \equivR$ in time polynomial in $|x| + |y|$. 
(2) For every $n$, $\equivR$ restricted  to $\Sigma^{\leq n}$ has at most polynomially (in $n$) many equivalence classes.

To relate parameterized and unparameterized problems, we employ \emph{cross-compositions}.  
Consider a language $L \subseteq \Sigma^*$ and a parameterized language $Q \subseteq \Sigma^* \times \Nat$. 
Then $L$ \emph{cross-composes} into $Q$ if there is a polynomial equivalence relation $\equivR$ and an algorithm $\kernel{A}$, referred to as the \emph{cross-composition}, 
with: 
$\kernel{A}$ takes as input a sequence $x_1, \dots, x_t \in \Sigma^*$ of strings that are equivalent with respect to $\equivR$, runs in time polynomial in $\Sigma_{i=1}^{t} \abs{x_i}$, and outputs an instance $(y,k)$ of $Q$ such that 
$k \leq p(\max_{i\in[1..t]}\abs{x_i} + \log(t))$ for a polynomial $p$. Moreover, $(y,k) \in Q$ if and only if there is a $i \in [1..t]$ such that $x_i  \in L$.
Cross-compositions are the key to lower bounds for kernels:
\begin{theorem}[\cite{Cygan2015}]\label{Theorem:NopolykernelGeneral}
	Assume that an $\NP$-hard language cross-composes into a parameterized language $Q$.
	Then $Q$ does not admit a polynomial kernel, unless \mbox{$\NP \subseteq \coNP/\poly$}.
\end{theorem} 

To show that $\BCS(\cs,\m)$ does not admit a polynomial kernel, we cross-compose $\kSAT{3}$ into $\BCS(\cs,\m)$.
Then Theorem \ref{Theorem:NopolykernelGeneral} yields the following:
\begin{theorem}\label{Theorem:Nopolykernel}
 $\BCS(\cs,\m)$ does not admit a polynomial kernel, unless $\NP \subseteq \coNP/\poly$.
\end{theorem}
\begin{proof}[Proof Idea]
	For the cross-composition, we first need a polynomial equivalence relation $\equivR$.
	Assume some standard encoding of $\kSAT{3}$-instances over a finite alphabet $\Gamma$.
	We let two encodings $\varphi, \psi$ be equivalent under $\equivR$ if both are proper $\kSAT{3}$-instances and have the same number of clauses and variables.
	
	Let $\varphi_1, \dots, \varphi_t$ be instances of $\kSAT{3}$ that are equivalent under $\equivR$.
	Then each $\varphi_i$ has exactly $\ell$ clauses and $k$ variables.
	We can assume that the set of variables is $\set{x_1, \dots, x_k}$.
	To handle the evaluation of these, we introduce the NFAs $A_i, i \in [1..k]$, each storing the value of $x_i$.
	We further construct an automaton $B$ that picks one out of the $t$ formulas $\varphi_j$.
	Automaton $B$ tries to satisfy $\varphi_j$ by iterating through the $\ell$ clauses.
	To satisfy a clause, $B$ chooses one out of the three variables and requests the corresponding value.
	
	The request by $B$ is synchronized with the memory $M$.
	After every such request, $M$ either ensures that the sent variable $x_i$ actually has the requested value or stops the computation. 
	This is achieved by a synchronization with the corresponding variable-automaton $A_i$, which keeps the value of $x_i$. 
	The number of context switches lies in $\bigO{\ell}$ and the size of the memory in $\bigO{k}$. 
	Hence, all conditions for a cross-composition are met.
\end{proof}

\section{Local Parameterization}
\label{Section:RR}
In the previous section, we considered a parameterization of $\BCS$ that was global in the sense that the threads had to share the number of context switches. 
We now study a parameterization that is \emph{local} in that every thread is given a budget of context switches. 

We would like to have a measure for the amount of communication between processes and consider only those computations in which heavily interacting processes are scheduled adjacent to each other.
The idea relates to \cite{LuPSY}, where it is observed that a majority of concurrency bugs already occur between a few interacting processes.

Given a word $u\in \Shuffle{i \in [1..t]}{} L(A_i)$, we associate with it a graph that reflects the order in which the threads take turns.
This \emph{scheduling graph} of $u$ is the directed multigraph $\csgraphof{u}=(V, E)$ with one node per thread that participates in $u$, $V \subseteq [1..t]$, and edge weights $E:V\times V\rightarrow \Nat$ defined as follows.
Value $E(i, j)$ is the number of times the context switches from thread $i$ to thread $j$ in $u$.
Formally, this is the number of different decompositions $u=u_1.a.b.u_2$ of $u$ so that $a$ is in the alphabet of $A_i$ and $b$ is in the alphabet of $A_j$.
Note that $E(i, i)=0$ for all $i\in[1..t]$.
In the following we will refer to directed multigraphs simply as graphs and distinguish between graph classes only where needed.

In the scheduling graph, the degree of a node corresponds to the number of times the thread has the processor.
The \emph{degree} of a node $n$ in $\mgraph=(V, E)$ is the maximum over the outdegree and the indegree, $\degreeof{n}=\maximum\set{\indegreeof{n}, \outdegreeof{n}}$.
As usual, the outdegree of a node $n$ is the number of edges leaving the node, $\outdegreeof{n}=\sum_{n'\in V} E(n, n')$, the indegree is defined similarly.
To see the correspondence, observe that a scheduling graph can have three kinds of nodes.
The \emph{initial node} is the only node where the indegree equals the outdegree minus 1, and the thread has the processor outdegree many times.
For the \emph{final node}, the outdegree equals the indegree minus 1, and the thread computes for indegree many contexts.
For all other (usual) nodes, indegree and outdegree coincide.
Any scheduling graph either has one initial, one final, and only usual nodes or, if initial and final node coincide, only consists of usual nodes.
The degree of the graph is the maximum among the node degrees, $\degreeof{G}=\maximum\Set{\degreeof{n}}{n\in V}$.

Our goal is to measure the complexity of schedules.
Intuitively, a schedule is simple if the threads take turns following some pattern, say round robin where they are scheduled in a cyclic way.
To formalize the idea of scheduling patterns, we iteratively contract scheduling graphs to a single node and measure the degrees of the intermediary graphs.
If always the same threads follow each other, we will be able to merge the nodes of such neighboring threads without increasing the degree of the resulting graph.
This discussion leads to a notion of scheduling dimension that we define in the following paragraph.
In Appendix \ref{Appendix:CarvingWidth}, we elaborate on the relation to an established measure: The carving-width.

Given a graph $\mgraph=(V, E)$, two nodes $n_1, n_2\in V$, and $n\notin V$, we define the operation of \emph{contracting $n_1$ and $n_2$} into the fresh node $n$ by adding up the incoming and outgoing edges.
Formally, the graph $\mgraph\contract{n_1, n_2}{n}=(V', E')$ is defined by \mbox{$V'=(V\setminus\set{n_1, n_2})\cup \set{n}$} and $E'(n', n) = E(n', n_1)+E(n', n_2)$, $E'(n, n') = E(n_1, n')+E(n_2, n')$, and \mbox{$E'(m, m') = E(m, m')$} for all other nodes.  
Using iterated contraction, we can reduce a graph to only one node.
Formally, a \emph{contraction process} of $\mgraph$ is a sequence $\contractionprocess=\mgraph_1, \ldots, \mgraph_{\abs{V}}$ of graphs, where $\mgraph_1 = \mgraph$, $\mgraph_{k+1}=\mgraph_{k}\contract{n_1, n_2}{n}$ for some $n_1, n_2\in V(G_{k})$ and $n\notin V(G_{k})$, $k\in[1..\abs{V}-1]$, and $\mgraph_{\abs{V}}$ consists of a single node.
The degree of a contraction process is the maximum of the degrees of the graphs in that process, $\degreeof{\contractionprocess}=\maximum\Set{\degreeof{\mgraph_i}}{i\in[1..\abs{V}]}$.
The \emph{scheduling dimension} of $\mgraph$ is $\sdimof{\mgraph}=\minimum\Set{\degreeof{\contractionprocess}}{\contractionprocess\text{ a contraction process of }\mgraph}$.

We study the complexity of $\BCS$ when parameterized by the scheduling dimension.
To this end, we define define the language of all computations where the scheduling dimension (of the corresponding scheduling graphs) is bounded by the parameter $\sdim\in \Nat$:
\vspace{-0.20cm}
\begin{align*}
	\sdimlang = \Set{u\in (\Sigma\times[1..t])^*}{\sdimof{\csgraphof{u}}\leq \sdim}.
\end{align*}
\vspace{-0.6cm}

\begin{problem}
	\problemtitle{Bounded Context Switching --- Local Parameterization}
	\problemshort{(\bcslsd)}
	\probleminput{$S = \SMS{\Sigma}{M}{A}{i}{t}$ and bound $\sdim\in \Nat$ on the scheduling dimension. }
	\problemquestion{Is $L(S)\cap \sdimlang \neq \emptyset$ ?}
\end{problem}
\begin{theorem}
	\bcslsd\ can be solved in time $\bigOS{(2m)^{4\sdim}4^{t}}$.
\end{theorem}

We present a fixed-point iteration that mimics the definition of contraction processes by
iteratively joining the interface sequences of neighboring threads.
Towards the definition of a suitable composition operation, let the \emph{product} of two interface sequences $\sigma$ \mbox{and $\tau$ be $\sigma\mergeop \tau = \bigcup_{\rho\in \sigma\SShuffle \tau} \rho\closure$}.
The language $\rho\closure$ consists of all interface sequences $\rho'$ obtained by
(iteratively) summarizing subsequences in $\rho$.
Summarizing $(r_1, r_1')\ldots (r_n, r_n')$ where \mbox{$r_1'=r_2$} up to $r_{n-1}'=r_n$ means to contract a sequence to $(r_1, r_n')$.
We write $\sigma\mergeop^k\tau$ for the variant of the product operation that only returns interface sequences of length at most \mbox{$k\geq 1$, $(\sigma\mergeop\tau) \cap (Q\times Q)^{\leq k}$.}

Our algorithm computes a fixed point over the powerset lattice (ordered by inclusion) $\Powerset(\ (Q\times Q)^{\leq \sdim}\times \Powerset([1..t])\ )$.
The elements are \emph{generalized interface sequences}, pairs consisting of an interface sequence together with the set of threads that has been used to construct it.
We generalize $\mergeop^k$ to this domain.
For the definition, consider $(\sigma_1, T_1)$ and $(\sigma_2, T_2)$.
If the sets of threads are not disjoint, $T_1\cap T_2\neq \emptyset$, the sequences cannot be merged and we obtain $(\sigma_1, T_1) \mergeop (\sigma_2, T_2) = \emptyset$.
If the sets are disjoint, we define \mbox{$(\sigma_1, T_1) \mergeoppar{k} (\sigma_2, T_2)= (\sigma_1\mergeop^k \sigma_2)\times\set{T_1\cup T_2}$}.
The fixed-point iteration is given by
\mbox{$L_1  = \bigcup_{i\in[1..t]}\ilangof{A_i}\times\set{\set{i}}$} and $L_{i+1} = L_i\cup (L_i\mergeop^{\sdim} L_i)$.
The following lemma states that it solves $\bcslsd$.
We elaborate on the complexity in Appendix \ref{Appendix:BCSL}.
\begin{lemma}\label{Lemma:lfp}
	\bcslsd\ holds iff the least fixed point contains $((q_{\mathit{init}}, q_{\mathit{final}}), T)$ for some $T$.
\end{lemma}
Problem $\bcslsd$ can be generalized and can be restricted in natural ways.
We discuss both options and show that variants of the above algorithm still apply, but yield different complexities.

Let problem $\bcslany$ be the variant of $\bcslsd$ where every thread is given a budget of running $\cs\in \Nat$ times, but where we do not make any assumption on the scheduling.
The observation is that, still, the scheduling dimension is bounded by $t\cdot \cs$.
The above algorithm solves $\bcslany$ in time $\bigOS{(2m)^{4 t\cdot\cs} 4^t}$.
\vspace{-0.3cm}
\subparagraph*{Fixing the Scheduling Graph.}
We consider $\bcslfix$, a variant of $\bcslsd$ where we fix a scheduling graph together with a contraction process of degree bounded by $\sdim$.
We are interested in finding an accepting computation that switches contexts as depicted by the fixed graph.
Formally, $\bcslfix$ takes as input an SMCP $S = \SMS{\Sigma}{M}{A}{i}{t}$, a scheduling graph $G$, and a contraction process $\contractionprocess$ of $G$ of degree at most $\sdim$.
The task is to find a word $u \in L(S)$ such that $G(u) = G$.
Our main observation is that a variant of the above algorithm applies and yields a runtime polynomial in $t$.
\begin{theorem}\label{Theorem:bcslfix}
	$\bcslfix$ can be solved in time $\bigOS{(2m)^{4\sdim}}$.
\end{theorem}

Fixing the scheduling graph $G = (V,E)$ and the contraction process $\contractionprocess$ has two crucial implications on the above algorithm.
First, we need to contract interface sequences with respect to the structure of $G$.
To this end, we will introduce a new product.
Secondly, instead of a fixed point we can now compute the required products between interface sequences iteratively along $\contractionprocess$.
Hence, we do not have to maintain the set of threads in the domain but can compute on $\Powerset((Q\times Q)^{\leq \sdim})$.

Towards obtaining the algorithm, we first describe the new product that summarizes interface sequences along the directed graph structure.
Let $\sigma$ and $\tau$ be interface sequences.
Further, let $\rho \in \sigma \SShuffle \tau$.
We call a position in $\rho$ an \emph{out-contraction} if it is of the form $(q_i,q_i')(p_j,p_j')$ so that $(q_i,q_i')$ belongs to $\sigma$, $(p_j,p_j')$ belongs to $\tau$ and $q_i' = p_j$.
Similarly, we define \emph{in-contractions}.
These are positions where a state-pair of $\tau$ is followed by a pair of $\sigma$.
The \emph{directed product} of $\sigma$ and $\tau$ is then defined as: $\sigma \omerge{i}{j} \tau = \bigcup_{\rho \in \sigma \SShuffle \tau} \rho \closure_{(i,j)}$.
Here, the language $\rho \closure_{(i,j)}$ contains all interface sequences $\rho'$ obtained by summarizing subsequences of $\rho$, in total containing exactly $i$ out-contractions and $j$ in-contractions.
Note that for $\sigma \in (Q \times Q)^{n}$ and $\tau \in (Q \times Q)^{k}$, the directed product contracts at $i+j$ positions and yields: 
$\sigma \omerge{i}{j} \tau \subseteq (Q\times Q)^{n + k - (i + j)}$. 

Now we describe the iteration. 
First, we may may assume that $V = [1..t]$.
Otherwise, the non-participating threads in $S$ can be deleted.
We distinguishes two cases.

In the first case, we assume that $G$ has a designated initial vertex $v_0$.
Then there is also a final vertex $v_f$.
Let $\contractionprocess = G_1, \dots, G_t$.
The iteration starts by assigning to each vertex $v \in V$ the set $S_v = \ilangof{A_v} \cap ({Q\times Q})^{\degreeof{v}}$.
For $S_{v_0}$, we further require that the first component of the first pair occurring in an interface sequence is $q_{\mathit{init}}$.
Similarly, for $S_{v_f}$ we require that the second component of the last pair is $q_{\mathit{final}}$.

Now we iterate along $\contractionprocess$:
For each contraction $G_{j+1} = G_{j}\contract{n_1,n_2}{n}$, we compute $S_{n} = (S_{n_1} \omerge{i}{k} S_{n_2})$, where $i = E(n_1,n_2)$ and $k = E(n_2,n_1)$. 
Then $S_n \subseteq (Q \times Q)^{\degreeof{n}}$, where $\degreeof{n}$ is the degree of $n$ in $G_{j+1}$.
Let $V(G_t) = \set{w}$.
Then the algorithm terminates after $S_{w}$ has been computed.

For the second case, suppose that no initial vertex is given.
This means that initial and final vertex coincide.
Then we iteratively go through all vertices in $V$, designate any to be initial (and final) and run the above algorithm.
The correctness is shown in the following lemma and we elaborate on the complexity in Appendix \ref{Appendix:BCSLFIX}.
\begin{lemma}\label{Lemma:fixedlfp}
	$\bcslfix$ holds iff $(q_{\mathit{init}} , q_{\mathit{final}}) \in S_{w}$.
\end{lemma}
\vspace{-0.6cm}
\subparagraph*{Round Robin.}
We consider an application of $\bcslfix$.
We define $\bcslrr$ to be the round-robin version of $\bcslsd$.
Again, each thread is given a budget of $\cs$ contexts, but now we schedule the threads in a fixed order: First thread $A_1$ has the processor, then $A_2$ is scheduled, followed by $A_3$ up to $A_t$.
To start a new round, the processor is given back to $A_1$.
The whole computation ends in $A_t$.
\begin{proposition}
	$\bcslrr$ can be solved in time $\bigOS{m^{4\cs}}$.
\end{proposition}
The problem $\bcslrr$ can be understood as fixing the scheduling graph to a cycle where every node $i$ is connected to $i+1$ by an edge of weight $\cs$ for $i \in [1..t-1]$ and the nodes $t$ and $1$ are connected by an edge of weight $\cs - 1$.
We can easily describe a contraction process: contract the vertices $1$ and $2$, then the result with vertex $3$ and up to $t$.
We refer to this as $\contractionprocess$.
Then we have $\degreeof{\contractionprocess} = \cs$.
Hence, we have constructed an instance of $\bcslfix$.

An application of the algorithm for $\bcslfix$ takes time at most $\bigOS{m^{4\cs}}$ in this case:
Let $G_{j+1} = G_{j}\contract{n_1,n_2}{n}$ be a contraction in $\contractionprocess$ with $j < t$.
Note that \mbox{$S_{n_1}, S_{n_2} \subseteq (Q \times Q)^{\cs}$}.
We have $E(n_1, n_2) = \cs$ and $E(n_2, n_1) = 0$.
Hence, the corresponding set $S_n$ is given by $(S_{n_1} \omerge{\cs}{0} S_{n_2}) \subseteq (Q \times Q)^{\cs}$.
Note that $\sigma \omerge{\cs}{0} \tau$ can be computed in linear time.
Similarly, for the last contraction $G_{t} = G_{t-1}\contract{n'_1,n'_2}{n'}$, where we have $S_{n'} = (S_{n_1'} \omerge{\cs}{\cs-1} S_{n_2'})$.
\vspace{-0.3cm}
\subparagraph*{Lower Bound for Round Robin.}
\label{Subsection:RRLower}
We prove the optimality of the algorithm for $\bcslrr$ by giving a reduction from $\kkClique$.
This variant of the classical clique problem asks for a clique of size $k$ in a graph whose vertices are elements of a $k \times k$ matrix.
Furthermore, the clique must contain exactly one vertex from each of the $k$ rows.
The problem was introduced as a part of the framework in \cite{Lokshtanov2011}.
It was shown that the brute-force approach is optimal:
$\kkClique$ cannot be solved in $2^{o(k \log k)}$ time, unless $\ETH$ fails.
We transport this to $\bcslrr$.
\begin{lemma}\label{Theorem:RRLB}
	Assuming $\ETH$, $\bcslrr$ cannot be solved in $2^{o(\cs \log ( m ) )}$ time.
\end{lemma}

\section{Discussion}
\label{Section:Discussion}

Our main motivation is to find bugs in shared memory concurrent programs.
In this setting, we can restrict our analysis to under-approximations:
We consider behaviors that are bounded in the number of context-switches, memory size or scheduling.
While this is enough to find bugs, there are cases where we need to check whether our program is actually correct.
We shortly outline circumstances under which we obtain an $\FPT$ upper bound, as well as a matching lower bound for the problem.

The reachability problem on a shared memory system in full generality is $\PSPACE$-complete. 
However, in real world scenarios, it is often the case that only a few (a fixed number of) threads execute in parallel with unbounded interaction. 
Thus, a first attempt is to parameterize the system by the number of threads $t$. 
But this yields a hardness result.
Indeed, the problem with $t$ as a parameter is hard for any level of the $\W$-hierarchy.

We suggest a parameterization by the number of threads $t$ and by $a$, the maximal size of the thread automata $A_{\threadid}$.
We obtain an $\FPT$-algorithm by constructing a product automaton.
The complexity is $\bigOS{a^t}$.
However, there is not much hope for improvement:
By a reduction from $\kkClique$, we can show that the algorithm is indeed optimal.

\bibliographystyle{plain}
\bibliography{content/cite}

\begin{thebibliography}{10}

\bibitem{Atig10}
M.~F. Atig.
\newblock Global model checking of ordered multi-pushdown systems.
\newblock In {\em {FSTTCS}}, pages 216--227. Schloss Dagstuhl--Leibniz-Zentrum
  fuer Informatik, 2010.

\bibitem{AtigBKS14}
M.~F. Atig, A.~Bouajjani, K.~N. Kumar, and P.~Saivasan.
\newblock On bounded reachability analysis of shared memory systems.
\newblock In {\em {FSTTCS}}, pages 611--623. Schloss Dagstuhl--Leibniz-Zentrum
  fuer Informatik, 2014.

\bibitem{AtigBT08}
M.~F. Atig, A.~Bouajjani, and T.~Touili.
\newblock Analyzing asynchronous programs with preemption.
\newblock In {\em {FSTTCS}}, pages 37--48. Schloss Dagstuhl--Leibniz-Zentrum
  fuer Informatik, 2008.

\bibitem{AtigBTA08}
M.~F. Atig, A.~Bouajjani, and T.~Touili.
\newblock On the reachability analysis of acyclic networks of pushdown systems.
\newblock In {\em {CONCUR}}, volume 5201, pages 356--371. Springer, 2008.

\bibitem{AtigBQ09}
M.F. Atig, A.~Bouajjani, and S.~Qadeer.
\newblock Context-bounded analysis for concurrent programs with dynamic
  creation of threads.
\newblock {\em Logical Methods in Computer Science}, 7(4), 2011.

\bibitem{BarnatBCMRS06}
J.~Barnat, L.~Brim, I.~Cern{\'{a}}, P.~Moravec, P.~Rockai, and O.~Simecek.
\newblock Divine - {A} tool for distributed verification.
\newblock In {\em {CAV}}, pages 278--281, 2006.

\bibitem{Biedl2013}
T.~C. Biedl and M.~Vatshelle.
\newblock The point-set embeddability problem for plane graphs.
\newblock {\em Int. J. Comput. Geometry Appl.}, 23(4-5), 2013.

\bibitem{Bjorklund2007}
A.~Bj{\"{o}}rklund, T.~Husfeldt, P.~Kaski, and M.~Koivisto.
\newblock Fourier meets m{\"{o}}bius: fast subset convolution.
\newblock In {\em {STOC}}, pages 67--74. ACM, 2007.

\bibitem{Bjorklund2016}
A.~Bj{\"{o}}rklund, P.~Kaski, and L.~Kowalik.
\newblock Constrained multilinear detection and generalized graph motifs.
\newblock {\em Algorithmica}, 74(2):947--967, 2016.

\bibitem{Bodlaender2009}
H.~L. Bodlaender, R.~G. Downey, M.~R. Fellows, and D.~Hermelin.
\newblock On problems without polynomial kernels.
\newblock {\em Journal of Computer and System Sciences}, 75(8):423--434, 2009.

\bibitem{BouajjaniEP11}
A.~Bouajjani, M.~Emmi, and G.~Parlato.
\newblock On sequentializing concurrent programs.
\newblock In {\em {SAS}}, pages 129--145. Springer, 2011.

\bibitem{Cai1997}
L.~Cai, J.~Chen, R.~G. Downey, and M.~R. Fellows.
\newblock On the parameterized complexity of short computation and
  factorization.
\newblock {\em Arch. Math. Log.}, 36(4-5):321--337, 1997.

\bibitem{Calabro2009}
C.~Calabro, R.~Impagliazzo, and R.~Paturi.
\newblock The complexity of satisfiability of small depth circuits.
\newblock In {\em {IWPEC}}, pages 75--85. Springer, 2009.

\bibitem{Cantin2005}
J.F. Cantin, M.H. Lipasti, and J.E. Smith.
\newblock The complexity of verifying memory coherence and consistency.
\newblock {\em IEEE Transactions on Parallel and Distributed Systems},
  16(7):663--671, 2005.

\bibitem{Cesati2003}
M.~Cesati.
\newblock The turing way to parameterized complexity.
\newblock {\em Journal of Computer and System Sciences}, 67(4):654--685, 2003.

\bibitem{Chen2006}
J.~Chen, X.~Huang, I.~A. Kanj, and G.~Xia.
\newblock Strong computational lower bounds via parameterized complexity.
\newblock {\em Journal of Computer and System Sciences}, 72(8):1346--1367,
  2006.

\bibitem{Cygan2016}
M.~Cygan, H.~Dell, D.~Lokshtanov, D.~Marx, J.~Nederlof, Y.~Okamoto, R.~Paturi,
  S.~Saurabh, and M.~Wahlstr{\"{o}}m.
\newblock On problems as hard as {CNF-SAT}.
\newblock {\em {ACM} Trans. Algorithms}, 12(3):41:1--41:24, 2016.

\bibitem{Cygan2015}
M.~Cygan, F.~V. Fomin, {\L}.~Kowalik, D.~Lokshtanov, D.~Marx, M.~Pilipczuk,
  M.~Pilipczuk, and S.~Saurabh.
\newblock Parameterized algorithms, 2015.

\bibitem{Cygan2011}
M.~Cygan, J.~Nederlof, M.~Pilipczuk, M.~Pilipczuk, J.~M.~M. van Rooij, and
  J.~O. Wojtaszczyk.
\newblock Solving connectivity problems parameterized by treewidth in single
  exponential time.
\newblock In {\em {FOCS}}, pages 150--159. IEEE Computer Society, 2011.

\bibitem{DemriLS02}
S.~Demri, F.~Laroussinie, and P.~Schnoebelen.
\newblock A parametric analysis of the state explosion problem in model
  checking.
\newblock In {\em STACS}, pages 620--631, 2002.

\bibitem{Downey2013}
R.~G. Downey and M.~R. Fellows.
\newblock {\em Fundamentals of Parameterized Complexity}.
\newblock Springer, 2013.

\bibitem{Durand-Gasselin15}
A.~Durand{-}Gasselin, J.~Esparza, P.~Ganty, and R.~Majumdar.
\newblock Model checking parameterized asynchronous shared-memory systems.
\newblock In {\em {CAV}}, 2015.

\bibitem{EneaF16}
C.~Enea and A.~Farzan.
\newblock On atomicity in presence of non-atomic writes.
\newblock In {\em TACAS}, pages 497--514, 2016.

\bibitem{EsparzaGM13}
J.~Esparza, P.~Ganty, and R.~Majumdar.
\newblock Parameterized verification of asynchronous shared-memory systems.
\newblock In {\em {CAV}}, 2013.

\bibitem{Esparza2014}
J.~Esparza, P.~Ganty, and T.~Poch.
\newblock Pattern-based verification for multithreaded programs.
\newblock {\em {ACM} Trans. Program. Lang. Syst.}, 36(3):9:1--9:29, 2014.

\bibitem{FaMa09}
A.~Farzan and P.~Madhusudan.
\newblock The complexity of predicting atomicity violations.
\newblock In {\em TACAS}, LNCS, pages 155--169. Springer, 2009.

\bibitem{FarzanM09}
Azadeh Farzan and P.~Madhusudan.
\newblock The complexity of predicting atomicity violations.
\newblock In {\em TACAS}, pages 155--169, 2009.

\bibitem{Fernau2015}
H.~Fernau, P.~Heggernes, and Y.~Villanger.
\newblock A multi-parameter analysis of hard problems on deterministic finite
  automata.
\newblock {\em J. Comput. Syst. Sci.}, 81(4):747--765, 2015.

\bibitem{Flum2006}
J.~Flum and M.~Grohe.
\newblock {\em Parameterized Complexity Theory (Texts in Theoretical Computer
  Science. An EATCS Series)}.
\newblock Springer, 2006.

\bibitem{Fomin2010}
F.~V. Fomin and D.~Kratsch.
\newblock {\em Exact Exponential Algorithms}, volume 111.
\newblock Springer, 2010.

\bibitem{Fomin2004}
F.~V. Fomin, D.~Kratsch, and G.~J. Woeginger.
\newblock Exact (exponential) algorithms for the dominating set problem.
\newblock In {\em WG}, pages 245--256. Springer, 2004.

\bibitem{FortinMW16}
M.~Fortin, A.~Muscholl, and I.~Walukiewicz.
\newblock On parametrized verification of asynchronous, shared-memory pushdown
  systems.
\newblock {\em CoRR}, abs/1606.08707, 2016.

\bibitem{Fortnow2011}
L.~Fortnow and R.~Santhanam.
\newblock Infeasibility of instance compression and succinct pcps for np.
\newblock {\em Journal of Computer and System Sciences}, 77(1):91--106, 2011.

\bibitem{Furbach2015}
F.~Furbach, R.~Meyer, K.~Schneider, and M.~Senftleben.
\newblock Memory-model-aware testing: A unified complexity analysis.
\newblock {\em ACM TECS}, 14(4):63:1--63:25, 2015.

\bibitem{Furer2009}
M.~F{\"{u}}rer.
\newblock Faster integer multiplication.
\newblock {\em {SIAM} J. Comput.}, 39(3):979--1005, 2009.

\bibitem{GibbonsKorach1997}
P.~B. Gibbons and E.~Korach.
\newblock Testing shared memories.
\newblock {\em SIAM Journal on Computing}, 26(4):1208--1244, 1997.

\bibitem{Hague11}
M.~Hague.
\newblock Parameterised pushdown systems with non-atomic writes.
\newblock In {\em {FSTTCS}}, pages 457--468. Schloss Dagstuhl--Leibniz-Zentrum
  fuer Informatik, 2011.

\bibitem{Harvey2016}
D.~Harvey, J.~van~der Hoeven, and G.~Lecerf.
\newblock Even faster integer multiplication.
\newblock {\em J. Complexity}, 36:1--30, 2016.

\bibitem{Impagliazzo2001}
R.~Impagliazzo and R.~Paturi.
\newblock On the complexity of k-sat.
\newblock {\em Journal of Computer and System Sciences}, 62(2):367--375, 2001.

\bibitem{TorreMP07}
S.~{La Torre}, P.~Madhusudan, and G.~Parlato.
\newblock A robust class of context-sensitive languages.
\newblock In {\em {LICS}}, pages 161--170. IEEE Computer Society, 2007.

\bibitem{TorreMP10}
S.~{La Torre}, P.~Madhusudan, and G.~Parlato.
\newblock Model-checking parameterized concurrent programs using linear
  interfaces.
\newblock In {\em {CAV}}, pages 629--644. Springer, 2010.

\bibitem{TorreN11}
S.~{La Torre} and M.~Napoli.
\newblock Reachability of multistack pushdown systems with scope-bounded
  matching relations.
\newblock In {\em {CONCUR}}, pages 203--218. Springer, 2011.

\bibitem{LalR08}
A.~Lal and T.~W. Reps.
\newblock Reducing concurrent analysis under a context bound to sequential
  analysis.
\newblock In {\em {CAV}}, pages 37--51. Springer-Verlag, 2008.

\bibitem{Lokshtanov2011}
D.~Lokshtanov, D.~Marx, and S.~Saurabh.
\newblock Slightly superexponential parameterized problems.
\newblock In {\em {SODA}}, pages 760--776. Society for Industrial and Applied
  Mathematics, 2011.

\bibitem{Lokshtanov2012}
D.~Lokshtanov, N.~Misra, S.~Saurabh, R.~G. Downey, F.~V. Fomin, and D.~Marx.
\newblock {\em Kernelization -- Preprocessing with a Guarantee}.
\newblock Springer, 2012.

\bibitem{LuPSY}
S.~Lu, S.~Park, E.~Seo, and Y.~Zhou.
\newblock Learning from mistakes: A comprehensive study on real world
  concurrency bug characteristics.
\newblock In {\em {ASPLOS}}, pages 329--339. ACM, 2008.

\bibitem{Marx07}
D.~Marx.
\newblock Can you beat treewidth?
\newblock In {\em FOCS}, 2007.

\bibitem{Marx2010}
D.~Marx.
\newblock Can you beat treewidth?
\newblock {\em Theory of Computing}, 6(1):85--112, 2010.

\bibitem{MusuvathiQ07}
M.~Musuvathi and S.~Qadeer.
\newblock Iterative context bounding for systematic testing of multithreaded
  programs.
\newblock In {\em {SIGPLAN}}, pages 446--455. ACM, 2007.

\bibitem{LeonFHM17}
H.~{Ponce de Le{\'{o}}n}, F.~Furbach, K.~Heljanko, and R.~Meyer.
\newblock Portability analysis for axiomatic memory models. {PORTHOS:} one tool
  for all models.
\newblock {\em CoRR}, abs/1702.06704, 2017.

\bibitem{QadeerR05}
S.~Qadeer and J.~Rehof.
\newblock Context-bounded model checking of concurrent software.
\newblock In {\em {TACAS}}, pages 93--107. Springer, 2005.

\bibitem{Seymour1994}
P.~D. Seymour and R.~Thomas.
\newblock Call routing and the ratcatcher.
\newblock {\em Combinatorica}, 14(2):217--241, 1994.

\bibitem{Bodlaender2005}
D.~M. Thilikos, M.~J. Serna, and H.~L. Bodlaender.
\newblock Cutwidth {I:} {A} linear time fixed parameter algorithm.
\newblock {\em J. Algorithms}, 56(1):1--24, 2005.

\bibitem{Rooij2009}
J.~M.~M. van Rooij, H.~L. Bodlaender, and P.~Rossmanith.
\newblock Dynamic programming on tree decompositions using generalised fast
  subset convolution.
\newblock In {\em {ESA}}, pages 566--577. Springer, 2009.

\bibitem{Wareham2000}
T.~Wareham.
\newblock The parameterized complexity of intersection and composition
  operations on sets of finite-state automata.
\newblock In {\em {CIAA}}, pages 302--310. Springer, 2000.

\bibitem{Yap1983}
C.~K. Yap.
\newblock Some consequences of non-uniform conditions on uniform classes.
\newblock {\em Theoretical Computer Science}, 26:287--300, 1983.

\end{thebibliography}

\newpage
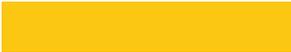
\appendix

\section{Proofs for Section~\ref{Section:Prelim}}
\label{Appendix:Prelim}

To prove Proposition~\ref{Theorem:BCSIntractable}, we begin by showing that
$\BCS(\cs,t)$ is a member of $\W[1]$.
This is achieved by a parameterized reduction from $\BCS(\cs,t)$ to the
$\W[1]$-complete problem \STMA~\cite{Cai1997}.
After that, we construct a parameterized reduction from $\BCS(\cs)$ to
$\BCS(\cs,t)$ and get that $\BCS(\cs)$ is actually in $\W[1]$.
For the hardness, we reduce from a parameterized \emph{intersection
	non-emptiness problem}, which is $\W[1]$-hard.
\begin{paraproblem}
	\problemtitle{Short Turing Machine Acceptance}
	\problemshort{(\STMA)}
	\probleminput{A nondeterministic Turing machine $\TM$,
		an input word $w$, and an integer $k \in \Nat$.}
	\problemparameter{$k$.}
	\problemquestion{Is there a computation of $\TM$ that accepts $w$
		in at most $k$ steps?}
\end{paraproblem}
\begin{lemma}
	We have $\BCS(\cs,t) \fptred \STMA$.
\end{lemma}
\begin{proof}
	Let $(S=\SMS{\Sigma}{M}{A}{i}{t}, \cs, t)$ be an instance of
	$\BCS(\cs, t)$ with shared memory $M = (Q, \Sigma, \delta, q_0, q_f)$
	and threads $A_i = (P_i, \Sigma \times \set{A_i}, \delta_i, p^0_i, p^f_i)$.
	We construct a nondeterministic Turing machine $\TM$ and a word $w$ so
	that $(\TM, w, (\cs+1)\cdot t + 2t)$ is an instance of $\STMA$ with the
	property: $L(S) \cap \Context(\Sigma,t,\cs) \neq \emptyset$ if and only if there is a computation
	of $\TM$ accepting $w$ in at most $(\cs+1) \cdot t + 2t$ steps.
	
	The idea behind $\TM$ is that the $i$-th cell of $\TM$'s tape stores the
	current state of $A_i$.
	The states of $M$ and a counter for the turns taken are represented by
	the control states of $\TM$.
	Moreover, the transition function of $\TM$ only allows steps which can
	be carried out simultaneously on $M$ and on one of the $A_i$.
	We want $\TM$ to work in three different \emph{modes}: a \emph{switch
		mode} to perform context switches, a \emph{work mode} to simulate runs
	of the $A_i$ and $M$ and an \emph{accept mode} which checks if $M$ and
	those $A_i$ that moved are in a final state.
	
	Formally, we set $\TM = (Q_\TM, \Gamma_\TM, \delta_\TM, \qinit, \qacc,
	\qrej)$, where:
	\begin{itemize}
		\item $\Gamma_\TM = \dot{\bigcup} P_i \cup \set{\fS_1, \dots,
			\fS_t, \fX, \$}$, the $\fS_i$ and $\fX$ are new letters and
		$\$$ is the \emph{left-end marker} of the tape,
		\item $Q_\TM = \set{\switch, \work} \times Q \times \set{0,
			\dots, \cs+2} \cup \set{\accept} \times \set{0, \dots, t}
		\cup \set{q_{acc}, q_{rej}}$ and
		\item $q_{init} = (\switch, q_0, 0)$.
	\end{itemize}
	
	\noindent Moreover, we set $w = \fS_1 \dots \fS_t$ and start $\TM$ on
	this word.
	Now we will explain the transition function $\delta_\TM$.
	Whenever $\TM$ is in \emph{switch}-mode, a new automaton $A_i$ is
	chosen to continue (or to start) the run.
	We allow walking left and right while remembering the current state
	of $M$ and the number of turns taken but without changing the tape
	content.
	So, for $q \in Q$ and $j \leq \cs$, we get:
	$$
	((\switch, q, j), p) \rightarrow ((\switch, q, j), p, D),
	$$
	where $D \in \set{L,R}$ and $p \in \Gamma_\TM \setminus \set{\fX}$. \\
	It is also possible to change the mode of $\TM$ to \emph{work}.
	In this case, we continue the run on the chosen automaton $A_i$.
	For $j \leq \cs$, we add:
	$$
	((\switch, q, j), p) \rightarrow ((\work, q, j), p, D).
	$$
	Once $\TM$ is in \emph{work mode}, there are two possibilities.
	Either the chosen automata $A_i$ did not move before, then there
	is an $\fS_i$ in the currently visited cell, or it has moved before,
	then the current state of $A_{i}$ is written in the cell.
	In the first case, we need a transition rule that \emph{activates}
	$A_{i}$ and does a first step.
	This step has to be synchronized with $M$. We get
	$$
	((\work, q, j), \fS_i) \rightarrow ((\switch, q', j+1), p', D),
	$$
	for all states $q,q' \in Q$ and $p' \in P_i$ so that $\sync{q}{q'}{i}
	{p^0_i}{p'} \neq \emptyset$.
	In the second case, we continue the run on $A_i$ while synchronizing
	with $M$:
	$$
	((\work, q, j), p) \rightarrow ((\switch, q', j+1), p', D),
	$$
	for all states $q,q' \in Q$ and $p, p' \in P_i$ so that $\sync{q}{q'}
	{i}{p}{p'} \neq \emptyset$. \\
	Note that after changing the mode from \emph{work} to \emph{switch},
	we know that a turn was taken and a context switch happened.
	To track this, $\TM$ increases its counter.
	This counter is not allowed to go beyond $\cs+1$.
	If this happens, $\TM$ will reject:
	$$
	((\switch, q, \cs+2), p) \rightarrow \qrej,
	$$
	for all $q \in Q$ and $p \in \Gamma_{\TM}$. \\
	If $\TM$ arrives in a state of the form $(\switch, q, i)$,
	where $q$ is a final state of $M$ and $i\in\set{1,\ldots,\cs+1}$,
	then it can enter \emph{accept mode}:
	$$
	((\switch, q, \cs+1), p) \rightarrow ((\accept, 0), p, D).
	$$
	Once $\TM$ is in \emph{accept mode}, it moves the head to the left
	end of the tape via additional moving transitions.
	Since we assume that the left end is marked by $\$, \TM$ can detect
	whether it reached the end.
	We get:
	$$
	((\accept, 0), p) \rightarrow ((\accept, 0), p, L),
	$$
	for $p \in \Gamma_M \setminus \set{\$}$ and
	$$
	((\accept, 0), \$)  \rightarrow ((\accept, 0), \$, R).
	$$
	After moving to the left end of the tape, $\TM$ will move right and
	if the current state of $A_i$, written in the $i$-th cell, is a final
	state, it gets replaced by $\fX$ and the counter, that counts the
	number of accepting automata, increases by $1$.
	If $\TM$ sees an $\fS_i$ in the $i$-th cell, it knows that $A_i$
	was never activated.
	This is also counted as accepting.
	We get:
	$$
	((\accept, j), p)  \rightarrow ((\accept, j+1), \fX, D),
	$$
	for $p$ a final state of one of the $A_i$ and
	$$
	((\accept, j), \fS_i) \rightarrow ((\accept, j+1), \fX, D).
	$$
	If $\TM$ reads $\fX$, it only moves left or right without changing
	the tape content or counter:
	$$
	((\accept, j), \fX) \rightarrow ((\accept, j), \fX, D).
	$$
	When $\TM$ detects $t$ accepting automata, then it will accept:
	$$
	((\accept, t), \fX) \rightarrow \qacc
	$$
	To simulate at most $\cs+1$ turns of the $A_i$, $\TM$ needs
	at most $(\cs+1) \cdot t$ steps.
	Once $\TM$ enters \emph{accept mode}, it needs at most $2t$ steps
	to verify that each $A_{i}$ is in a final state or did not move at
	all.
	Hence, we are looking for computations of length at most $(\cs+1)
	\cdot t + 2t$.
	It is easy to observe that the reduction works correctly and can be
	constructed in polynomial time.
\end{proof}
\newpage
\begin{lemma}
	We have $\BCS(\cs) \fptred \BCS(\cs,t)$.
\end{lemma}
\begin{proof}
	Let $(S=\SMS{\Sigma}{M}{A}{i}{t}, \cs)$ be an instance of $\BCS(\cs)$ with
	$M = (Q, \Sigma, \delta, q_0, q_f)$ and $A_i = (P_i, \Sigma \times \set{A_i},
	\delta_i, p^0_i, p^f_i)$.
	To construct an instance of $\BCS(\cs,t)$, the rough idea is that in at most
	$\cs$ context switches, we can use at most $\cs+1$ different automata.
	Hence, we introduce $\cs+1$ new finite automata, where each chooses to
	simulate one of the $A_j$.
	
	We set $\Gamma = \Sigma \cup \set{\#_1, \dots, \#_t} \cup \set{\square}$
	and define automaton $B_i = (P'_i,\Gamma \times \set{B'_i}, \delta'_i,
	\pchoice{i}, F'_i)$ for $i \in [1..\cs+1]$, where:
	\begin{itemize}
		\item $P'_i = \dot{\bigcup}_{j=1}^{t} P_j \cup \set{ \pchoice{i},
			\plock{i} }$ and $\pchoice{i}, \plock{i}$ are new states, and
		\item $F'_i = \dot{\bigcup}_{j = 1}^{t} \set{p^f_j} \cup
		\set{\plock{i}}$.
	\end{itemize}
	
	\noindent The transition relation $\delta'_i$ contains all transition
	rules of the $A_j$: if $p \Move{a} p'$ is an edge in $A_j$, we get an
	edge $p \Move{a} p'$ in $B_i$.
	Moreover, we add rules $\pchoice{i} \xrightarrow{\#_j} p^0_j$ for
	$j \in[1..t]$, and we add \mbox{$\pchoice{i} \xrightarrow{\square} \plock{i}$.}
	An illustration of $B_i$ is given in Figure~\ref{Figure:BSPkToBSPkn_Bi}.
	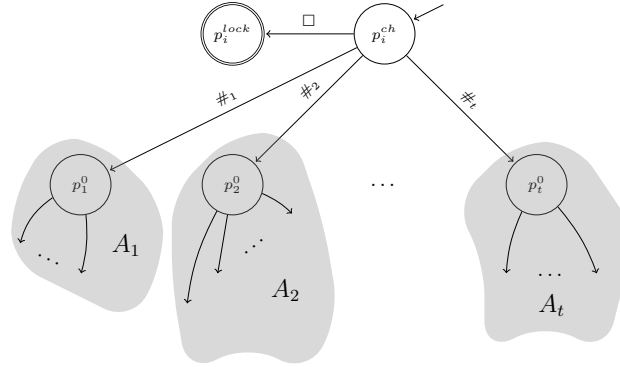
\begin{figure}[h]
	\centering

	\scalebox{0.8}{
	\begin{tikzpicture}[shorten >= 1pt, node distance = 2.5cm]

		\node[state, minimum height = 1cm] (pch) {\scriptsize $\pchoice{i}$};
		\draw[<-, shorten <= 1pt] (pch) -- ++ (1,0.5);
		\node[state, accepting, minimum height = 1cm, left of = pch] (plock) {\scriptsize $\plock{i}$};

		\path[->] (pch) edge node[pos = 0.5, above] {\scriptsize $\square$} (plock);

		\node[state, minimum height = 1cm, left of = plock, below of = plock] (p1init) {\scriptsize $p^0_1$};
		\node[state, minimum height = 1cm, right of = p1init] (p2init) {\scriptsize $p^0_2$};
		\node[right of = p2init] (dots) {$\dots$};
		\node[state, minimum height = 1cm, right of = dots] (pninit) {\scriptsize $p^0_t$};

		\path[->] (pch) edge node[pos = 0.5, above, sloped] {\scriptsize $\#_{1}$} (p1init);
		\path[->] (pch) edge node[pos = 0.4, above, sloped] {\scriptsize $\#_{2}$} (p2init);
		\path[->] (pch) edge node[pos = 0.5, above, sloped] {\scriptsize $\#_{t}$} (pninit);

		\coordinate[yshift = 10pt] (A11) at (p1init.north east);
		\path[-, fill = gray, opacity = 0.3, rounded corners = 15pt]
			(A11) --
			++ (1,-1) --
			++ (0,-1) --
			++ (-0.5,-1) --
			++ (-2,1) --
			++ (0,1) --
			++ (1,1.5,) --
			cycle;

		\coordinate[xshift = -1cm, yshift = -1cm] (A1post1) at (p1init);
		\coordinate[yshift = -1.5cm] (A1post2) at (p1init);
		\path[->] (p1init) edge[bend right = 20] (A1post1);
		\path[->] (p1init) edge[bend left = 10] (A1post2);
		\path (A1post1) -- node[pos = 0.5, sloped] (A1dots) {$\dots$} (A1post2);

		\node[xshift = 0.75cm, yshift = 0.5cm] at (A1post2) {\Large $A_{1}$};

		\coordinate[yshift = 10pt] (A21) at (p2init.north);
		\path[-, fill = gray, opacity = 0.3, rounded corners = 15pt]
			(A21) --
			++ (-1,-1) --
			++ (0,-1.5) --
			++ (0.25,-1.5) --
			++ (1,0.5) --
			++ (1,-0.5) --
			++ (0.5,1) --
			++ (-0.25,1.5) --
			++ (-0.5, 1.5) --
			cycle;

		\coordinate[xshift = -0.75cm, yshift = -2cm] (A2post1) at (p2init);
		\coordinate[xshift = -0.25cm, yshift = -1.5cm] (A2post2) at (p2init);
		\coordinate[xshift = 1cm, yshift = -0.5cm] (A2post3) at (p2init);
		\path[->] (p2init) edge[bend right = 10] (A2post1);
		\path[->] (p2init) edge (A2post2);
		\path[->] (p2init) edge[bend left = 10] (A2post3);
		\path (A2post2) -- node[pos = 0.5, sloped] (A2dots) {$\dots$} (A2post3);

		\node[xshift = 0.5cm, yshift = -0.75cm] at (A2dots) {\Large $A_{2}$};

		\coordinate[yshift = 10pt] (An1) at (pninit.north west);
		\path[-, fill = gray, opacity = 0.3, rounded corners = 15pt]
			(An1) --
			++ (-1,-1) --
			++ (0.5,-1) --
			++ (0,-1.5) --
			++ (1,0.5) --
			++ (1,-0.5) --
			++ (0.5,2) --
			++ (-1,1.5) --
			cycle;

		\coordinate[xshift = -0.5cm, yshift = -1.5cm] (Anpost1) at (pninit);
		\coordinate[xshift = 1cm, yshift = -1.5cm] (Anpost2) at (pninit);
		\path[->] (pninit) edge[bend right = 10] (Anpost1);
		\path[->] (pninit) edge[bend left = 10] (Anpost2);
		\path (Anpost1) -- node[pos = 0.5, sloped] (Andots) {$\dots$} (Anpost2);

		\node[yshift = -0.5cm] at (Andots) {\Large $A_{t}$};

	\end{tikzpicture}
	}

	\caption{Automaton $B_i$ can either choose to simulate one of the $A_j$ or not to simulate any of the $A_j$. To this end, it keeps a copy of all $A_j$ and a deadlock state that can be accessed via writing $\square$.}
	\label{Figure:BSPkToBSPkn_Bi}
\end{figure}

	Now we define the finite automaton $N$ to be the tuple $(Q', \Gamma, \delta',
	\qchoice, q_f)$, with \mbox{$Q' = Q \cup \Set{q_{(\#_j, i)}} {j \in [1..t] \text{ and } i \in[1..cs]} \cup \set{\qchoice}$.}
	The transition relation $\delta'$ is the union of the transition relation of
	$M$ and the rules explained below.
	To force each $B_i$ to make a choice, we add transitions $q_{(\#_j, i)}
	\Move{\#_l} q_{(\#_l, i+1)}$ for $j \in[1..t]$, $j < l \leq t$ and
	$i \in [1..\cs-1]$.
	Note that we assume that the choice of the $A_l$ is done in increasing order.
	This prevents the $B_i$ from choosing the same $A_l$.
	For the initial choice, we add transitions $\qchoice \Move{\#_l} q_{(\#_l, 1)}$
	for all $l \in [1..t]$.
	For the final choice, we add the rules $q_{(\#_j,\cs)} \Move{\#_l} q_0$  for
	$j \in [1..t]$ and $j < l \leq t$.
	To give the $B_i$ the opportunity not to simulate any of the $A_l$, we add
	the transitions $q_{(\#_j,i)} \Move{\square} q_{(\#_j,i+1)}$ for $j \in [1..t]$
	and $i \in [1..\cs-1]$ and $q_{(\#_j,\cs)} \Move{\square} q_0$ for all
	$j \in [1..t]$.
	An image of automaton $N$ can be found in Figure~\ref{Figure:BSPkToBSPkn_N}.
	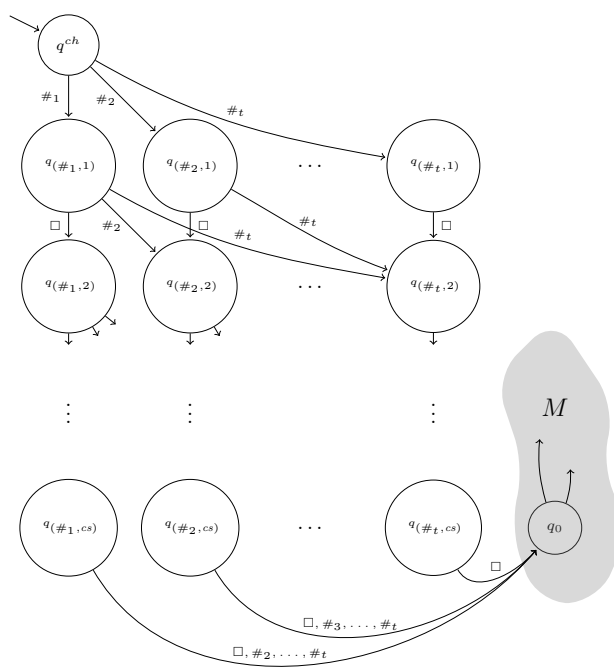
\begin{figure}[h]
	\centering

	\scalebox{0.8}{
	\begin{tikzpicture}[shorten >= 1pt, node distance = 1.5cm]

	\node[state, minimum height = 1cm] (qch) {\scriptsize $q^{ch}$};
	\draw[<-, shorten <= 1pt] (qch) -- ++ (-1,0.5);

	\coordinate (start) at (qch);

	\foreach \i in {1,2,3,4}{
		\foreach \j in {1,2,3,4}{
			\ifthenelse{\i = 3}{
				\ifthenelse{\j=3}{
					\node (r\i\j) at ($(start) + (\i*2-2, -\j*2)$) {};
				}{
					\node (r\i\j) at ($(start) + (\i*2-2, -\j*2)$) {$\dots$};
				}
			}{
				\ifthenelse{\j = 3}{
					\node (r\i\j) at ($(start) + (\i*2-2, -\j*2)$) {$\vdots$};
				}{
					\node[state] (r\i\j) at ($(start) + (\i*2-2, -\j*2)$) {
						\ifthenelse{\i = 4}{
							\ifthenelse{\j=4}{
								\tiny$q_{(\#_t,\cs)}$
							}{
								\tiny $q_{(\#_t,\j)}$
							}
						}{
							\ifthenelse{\j=4}{
								\tiny $q_{(\#_{\i},\cs)}$
							}{
								\tiny $q_{(\#_{\i},\j)}$
							}
						}
					};
				}
			}
		}
	}

	\path[->] (qch) edge node[left] {\tiny $\#_{1}$} (r11);
	\path[->] (qch) edge node[left] {\tiny $\#_{2}$} (r21);
	\path[->] (qch) edge[out = 330, in = 170] node[above] {\tiny $\#_{t}$} (r41);

	\path[->] (r11) edge node[left] {\tiny $\square$} (r12);
	\path[->] (r11) edge node[left] {\tiny $\#_{2}$} (r22);
	\path[->] (r11) edge[out = 330, in = 170] node[above] {\tiny $\#_{t}$} (r42);

	\path[->] (r21) edge node[right] {\tiny $\square$} (r22);
	\path[->] (r21) edge[out = 330, in = 160] node[above] {\tiny $\#_{t}$} (r42);

	\path[->] (r41) edge node[right] {\tiny $\square$} (r42);

	\path[->, draw] (r12) -- ++ (0,-1);
	\path[->, draw] (r12) -- ++ (0.5,-0.85);
	\path[->, draw] (r12) -- ++ (0.8,-0.65);

	\path[->, draw] (r22) -- ++ (0,-1);
	\path[->, draw] (r22) -- ++ (0.5,-0.85);

	\path[->, draw] (r42) -- ++ (0,-1);

	\node[state, right of = r44, node distance = 2cm] (qinit) {\scriptsize $q_0$};

	\path[fill = gray, opacity = 0.3, rounded corners = 15pt]
		($(qinit) + (-0.75,0)$) --
		++ (0.25,1) --
		++ (-0.5,1.5) --
		++ (1,1) --
		++ (1,-1.5) --
		++ (-0.25,-1) --
		++ (0.5,-1) --
		++ (-0.55,-1.25) --
		++ (-1,0) --
		cycle;

	\path[->] (qinit) edge[bend left = 10] ++ (-0.25,1.5);
	\path[->] (qinit) edge[bend right = 10] ++ (0.25,1);
	\node[yshift = 2cm] at (qinit) {\Large $M$};

	\path[->] (r44) edge[out = 300, in = 230] node[above] {\tiny $\square$} (qinit);
	\path[->] (r24) edge[out = 300, in = 230] node[pos = 0.45, above] {\tiny $\square, \#_{3}, \dots, \#_{t}$} (qinit);
	\path[->] (r14) edge[out = 300, in = 230] node[pos = 0.45, above] {\tiny $\square, \#_{2}, \dots, \#_{t}$} (qinit);

	\end{tikzpicture}
	}

	\caption{First, automaton $N$ checks whether the $B_{i}$ have chosen at most $\cs+1$ different automata $A_{j}$ to simulate, then it starts the simulation of $M$.}
	\label{Figure:BSPkToBSPkn_N}
\end{figure}

	For the choice of the $B_i$, we do exactly $\cs$ context switches.
	Then we need at most another $\cs$ context switches for the simulation of
	the chosen $A_l$ and $M$.
	Hence, the tuple given by $(S'=\SMS{\Gamma}{N}{B}{i}{\cs+1}, 2\cs)$ is an
	instance of $\BCS(\cs, t)$ and it is easy to see that we have:
	$$
	L(S') \cap \Context(\Gamma, \cs+1, 2\cs) \neq \emptyset \text{ if and only if } L(S) \cap \Context(\Sigma, t, \cs) \neq \emptyset.
	$$
\end{proof}
We make use of this bounded version of the \emph{intersection non-emptiness problem}.
\begin{problem}
	\problemtitle{Bounded DFA Intersection Non-emptiness}
	\problemshort{(\BDFAI)}
	\probleminput{Deterministic finite automata $B_1, \dots, B_n$ over an alphabet $\Sigma$ and an integer $m \in \Nat$.}
	\problemquestion{Does there exist a word $w$ of length $m$ so that $w \in \bigcap_{i=1}^{n} L(B_i)$?}
\end{problem}
We parameterize by the length $m$ of the word that we are seeking for and by $n$, the number of involved automata.
We refer to the problem as $\BDFAI(m,n)$ and it is know that $\BDFAI(m,n)$ is $\W[1]$-complete~\cite{Cesati2003, Wareham2000}.
\begin{lemma}\label{Lemma:BDFAI}
	We have $\BDFAI(m,n) \fptred \BCS(cs)$.
\end{lemma}
\begin{proof}
	Let $(B_1, \dots, B_n, m)$ be an instance of $\BDFAI(m,n)$ over the
	alphabet $\Gamma$.
	We construct an instance $(S=\SMS{\Sigma}{M}{A}{i}{n}, m \cdot n)$ of
	$\BCS(\cs)$ so that
	$$
	L(S) \cap \Context(\Sigma, n, m \cdot n) \neq \emptyset \text{ if and only if } \Sigma^{m} \cap
	\bigcap_{i=1}^{n} L(B_{i}) \neq \emptyset.
	$$
	Set $\Sigma = \Gamma \times \set{1, \dots, n}$.
	We construct an automaton $A_i$ which simulates $B_i$ on $\Gamma$.
	To this end, $A_i$ will have the states of $B_i$ and for each transition
	$p_i \Move{a} p'_i$ of $B_i$, we get a transition \mbox{$p_i \Move{(a,i)} p'_i$} in $A_i$.
	Let $M$ be an automaton accepting the language \mbox{$(\Set{(a,1)\ldots(a,n)}{a\in\Sigma})^m$.}
	This ensures that each $B_i$ reads the same letter and that we only get words in the intersection.
	Clearly, the reduction works correctly.
\end{proof}
\section{Proofs for Section~\ref{Section:BCS}}
\label{Appendix:BCS}
\subsection{Upper Bounds}
\begin{proof}[Proof of Lemma~\ref{Lemma:MemorySequenceEquivalence}]
	To prove part (i), we show two inclusions.
	For the first inclusion, let \mbox{$w = (q_0, q_1) \dots (q_{m-1}, q_m)$} be an induced interface sequence, an element in $\iilangof{S}$.
	Then there is a word $u \in L(S)$ that induces $w$.
	This means, that we can write $u$ as $u = u_1 \dots u_m$ and there is an accepting run $r$ of $M$ on $u$ of the form:
	\begin{align*}
		q_0 \xrightarrow{u_1} q_1 \xrightarrow{u_2} q_2 \ldots
		q_{m-1} \xrightarrow{u_{m}} q_m.
	\end{align*}
	Since $u$ also lies in the shuffle of the $L(A_i)$, there are subwords $u^1, \ldots, u^t$, forming a partition of $u$ such that $u^i\in L(A_i)$.
	Following run $r$, every subword $u^i$ leads to an (possibly non-valid) interface sequence
	$w^i = (q_{i_1},q_{i_2})\ldots(q_{i_l},q_{i_{l+1}})$. 
	These partition $w$ and by construction, we get that $w^i \in \ilangof{A_i}$.
	Thus, $w \in \Shuffle{i \in [1..t]}{} \ilangof{A_i}$ and clearly, $w$ is valid.
	
	For the converse inclusion, let $w$ be a valid sequence in $\Shuffle{i \in [1..t]}{} \ilangof{A_i}$.
	Then there are subsequences $w^i \in \ilangof{A_i}$, forming a partition of $w$.
	By construction of $\ilangof{A_i}$, for each $w^i$ there is a word $u^i \in L(A_i)$ that follows the state changes in $M$ depicted by $w^i$.
	We compose (shuffle) the $u^i$ in the same order as the $w^i$ compose to $w$.
	Hence, we get a word $u$ in the shuffle of the $L(A_i)$. 
	Since $w$ is valid, $u$ follows the states in $M$ given by $w$ and thus, lies in $L(M)$.
	This implies: $w$ is induced by $u$. \\
	
	To show the second part of Lemma \ref{Lemma:MemorySequenceEquivalence}, we construct a finite automaton for the language $\ilangof{A_{\threadid}}$.
	We define $B_{\threadid}$ over the alphabet $Q \times Q$.
	The states are the states of $A_{\threadid}$.
	We add a transition from $p$ to $p'$ in $B_{\threadid}$ labeled by $(q,q')$ if $\sync{q}{q'}{\threadid}{p}{p'} \neq \emptyset$.
	Then, clearly $L(B_{\threadid}) = \ilangof{A_{\threadid}}$.
	Computing whether all the intersections $\sync{q}{q'}{i}{p}{p'}$ are non-empty can be done in $\bigO{\abs{A_{\threadid}}^3 \cdot \abs{M}^3}$.
\end{proof}
\subsection{Lower Bounds}
\subparagraph*{Lower Bound for Bounded Context Switching}
\begin{proof}[Proof of Theorem \ref{Theorem:BCSLower}]
	For the reduction, let two graphs $G$,$H$ be given, let \mbox{$k = \abs{E(G)}$} be  the number of edges, and \mbox{$\set{e_1,\dots,e_l} = V(G)$} be the vertices of $G$.
	Isolated vertices are not relevant for the complexity of $\SGI$, hence we assume there are none in $G$, which gives $l \leq 2k$.
	
	We will construct an instance $(S = \SMS{\Sigma}{M}{A}{i}{k}, 2k)$ to $\BCS(\cs)$. 
	To this end, we set $\Sigma = V(G) \times V(H)$, so intuitively each letter describes a map of a vertex of $G$ to a vertex of $H$.
	Following this intuition we will use automaton $M$ to output all possible mappings of $V(G)$ to $V(H)$, and each $A_i$ to verify that the $i$-th edge of $G$ is mapped to and edge of $H$.
	
	Pick any order $\prec$ on $V(H)$.
	We let $M$ accept the language $(v_1,w_{1})^{d_1}(v_2,w_{2})^{d_2}\dots(v_l,w_{l})^{d_l},$ 
	where $v_i\in V(G)$, $w_i\in V(H)$ and $\sum_{i=1}^l d_i=2k$, and $w_i \prec w_j$ for all $1\leq i<j\leq l$.
	Note that the order is needed to avoid that different vertices of $G$ get mapped to one \mbox{vertex in $H$}.
	
	For each edge of $G$ we will have an automata $A_i$.
	For $e_i=(v_{s},v_{t})$, we let $A_i$ accept the language 
	$$\bigcup_{(w_s,w_t)\in E(H)}\begin{cases}
	(v_s,w_s)(v_t,w_t) & \text{ if $w_s\prec w_t$, } \\
	(v_t,w_t)(v_s,w_s) & \text{ else.}\end{cases}$$
	
	We show that $\BCS$ has a solution with at most $2k$ context switches if and only if $G$ is isomorphic to a subgraph of $H$.
	
	First, let $\BCS$ have a solution, i.e., we find a word $u$ in $L(M)$ and in the shuffle of a subset of the languages $L(A_1), \dots, L(A_k)$. 
	Since the word is in $L(M)$, we know it is of the form $u = (v_1,w_{1})^{d_1}(v_2,w_{2})^{d_2} \dots (v_l,w_{l})^{d_l}$. 
	Since each $A_i$ accepts only a word of length $2$ and this word has length $2k$, each $A_i$ is involved in the shuffle.
	By our assumption that $G$ has no singletons, each vertex $v\in V(G)$ is incident to at lease on edge $e_j\in E(G)$. 
	Since $A_j$ is part of the shuffle, the input has to contain one letter of $\set{v} \times V(H)$.
	As there are exactly $l$ different letters in the word, we can define a map $\psi: V(G) \rightarrow V(H)$ where $\psi(v_i) = w_i$ for $1\leq i\leq l$. 
	Assume that $e_i=(v_s,v_t) \in E(G)$ and $(\psi(v_s),\psi(v_t)) \notin E_(H)$.
	Then $A_i$ would not accept any subword $u$, a contradiction to the fact that all $A_i$ are involved in the shuffle.
	Hence each edge of $G$ is mapped to an edge of $H$ and $\psi$ is an isomorphism, embedding $G$ into $H$.
	
	Now let $\psi$ be an embedding, mapping $G$ isomorphic to a subgraph of $H$.
	We order $v_1, \dots, v_l \in V(G)$ such that $\psi(v_i)\prec\psi(v_j)$ for all $1\leq i<j\leq l$. 
	Then one can directly see that $(v_1,\psi(v_1))^{d_1}(v_2,\psi(v_2))^{d_2}\dots(v_l,\psi(v_l))^{d_l}$, where $d_i$ is the degree of $v_i$, is in $L(M)$.
	It is also in the shuffle of $L(A_1),\dots,L(A_k)$ as each edge $(v_s, v_t) \in E(G)$ is mapped to an edge $(\psi(v_s),\psi(v_t))\in E(H)$.
	It remains to show that we have at most $2k$ context switches, but this is clear as the word length is $2k$.
	
	Finally we need to show that the reduction can be compute in polynomial time.
	To this end, we have to show that the size of the automata $M$ and $A_i$ is polynomially bounded.
	The number of states of $M$ is bounded by $\abs{V(G)} \cdot \abs{V(H)} \cdot \abs{V(G)}\cdot\abs{E(H)}$ as the automata only need to remember the last letter, the number of different letters produced and the word length.
	Each of the small automata need $\abs{V(H)} + 2$ states as it only needs to remember the vertex of $H$ read in the first letter.
\end{proof}
\subparagraph*{Lower Bound for Shuffle Membership}
For the proof of Proposition \ref{Theorem:LowerBoundSM} we make use of a further result, explicitly stated as \mbox{Theorem $6$} in \cite{Bjorklund2016}, implicitly as Theorem $4.7$ in \cite{Cygan2016}:
\begin{theorem}
	If $\SetCov$ can be solved in $\bigOS{(2-\delta)^{n+t}}$ time for a $\delta > 0$ then it can also be solved in $\bigOS{(2-\varepsilon)^{n}}$ time, for an $\varepsilon > 0$.
\end{theorem}

This allows us to reason as follows.
Assume we have a polynomial-time reduction from $\SetCov$ to $\SM$ such that an instance $((S_i)_{i \in [1..m]},t)$ of $\SetCov$ is mapped to an instance $((B_i)_{i\in[1..m]}, k, w)$ of $\SM$ with $k = n+t$.
Then an \mbox{$\bigOS{(2 - \delta)^k}$}-time algorithm for $\SM$ would yield an $\bigOS{(2 - \varepsilon)^n}$-time algorithm for $\SetCov$.
Hence, all what is left to complete the proof of Proposition \ref{Theorem:LowerBoundSM} is to construct such a reduction.
This is done in the following lemma:
\begin{lemma}\label{Lemma:SCtoSM}
	There is a polynomial-time reduction from $\SetCov$ to $\SM$ such that an instance $((S_i)_{i \in [1..m]},t)$ of $\SetCov$ is mapped to an instance $((B_i)_{i\in[1..m]}, k, w)$ of $\SM$ with $k = n+t$.
\end{lemma}
\begin{proof}
	Let $((S_i)_{i \in [1..m]},t)$ be an instance of $\SetCov$ with \mbox{$U = \set{ u_1, \dots, u_n }$}.
	We construct $\Gamma = U \cup \set{ 1, \dots, t }$ and introduce an NFA $B_S$ for each set $S$ in the given family.
	The automaton $B_S$ has two states and its language is \mbox{$L(B_S) = \Set{ u^*.j }{ u \in S, j \in [1..t] }$}.
	We further define the word $w$ to be the concatenation of the two words $w_U = u_1 \dots u_n$ and $w_t = 1 \dots t$.
	Hereby, $w_U$ ensures that each element of $U$ gets covered while $w_t$ ensures that we use exactly $t$ sets.
	Note that the length of $w$ is $n+t$.
	Hence, we constructed an instance $((B_i)_{i\in[1..m]}, n+t, w)$ of $\SM$.
	
	For the correctness of the above construction, first assume that $U$ can be covered by $t$ sets of the given family.
	After reordering, we may assume that $S_1, \dots, S_t$ cover $U$.
	Now we can use an interleaving of the $B_{S_i}, i \in [1..t]$ to read $w_U$: Each $B_{S_i}$ reads those $u \in U$ that get covered by $S_i$.
	Note that an element $u$ can lie in more than one of the $S_i$.
	In this case, $u$ is read non-deterministically by one of the corresponding $B_{S_i}$.
	After reading the elements of $U$, append the index $i$ to the string read by $B_{S_i}$.
	Hence, we get that $w$ can be read by interleaving the $B_{S_i}$ with at most $n+t-1$ context switches, $w \in \SShuffle_{i \in [1..m]}^{} L(B_{S_i})$.
	
	Now let $w = w_U.w_t$ be in the shuffle of the $L(B_{S_i}), i \in [1..m]$.
	Since $w_t = 1 \dots t$, we get that exactly $t$ of the automata $B_{S_i}$ are used to read the word $w$.
	We may assume that these are $B_{S_1}, \dots, B_{S_t}$.
	Then the prefix $w_U$ is read by interleaving the $B_{S_i}$.
	This means that each $u \in U$ lies in (at least) one of the $S_i$ and hence, $S_1, \dots, S_t$ cover the universe $U$.
\end{proof}
\subparagraph*{Lower Bound on the Size of the Kernel}
\begin{proof}[Proof of Theorem \ref{Theorem:Nopolykernel}]
	First, we define the polynomial equivalence relation.
	We assume that the $\kSAT{3}$-instances are encoded over a finite alphabet $\Gamma$.
	Let $\Formulas$ denote the set of encodings that actually encode proper $\kSAT{3}$-instances.
	Let $\varphi, \psi$ be two encodings from $\Gamma^*$.
	We define $(\varphi, \psi) \in \equivR$ if and only if (1) $\varphi, \psi \in \Formulas$ and they have the same number of clauses and variables, or (2) both, $\varphi$ and $\psi$, do not lie in $\Formulas$.
	Note that the relation $\equivR$ meets all the requirements on a polynomial equivalence relation.
	
	Now we elaborate on the cross-composition.
	Let $\varphi_1, \dots, \varphi_t$ be instances of $\kSAT{3}$, equivalent with respect to $\equivR$.
	This means that all given formulas have the same number $\ell$ of clauses and $k$ variables.
	We may assume that the set of variables used by any of the $\varphi_j$ is $\set{x_1, \dots, x_k}$.
	
	We start constructing the needed shared memory concurrent program $S$ by defining the underlying alphabet: $\Sigma = (\set{x_1, \dots, x_k} \times \set{ ?0, !0, ?1, !1}) \cup \set{\#}$.
	Intuitively, $(x_i, ?0)$ corresponds to querying if variable $x_i$ evaluates to $0$ and $(x_i, !0)$ corresponds to verifying that  $x_i$ indeed evaluates to $0$.
	The symbol $\#$ was added to prevent that the empty word lies in $L(S)$.
	
	For each variable $x_i$, we introduce an NFA $A_i$ that keeps track of the value assigned to $x_i$.
	To this end, it has three states: An initial state and two final states, one for each possible value. 
	The automaton $A_i$ accepts the language $(x_i, !0)^+ + (x_i, !1)^+$.
	
	We further introduce a thread $B$, responsible for checking whether one out of the $t$ given formulas is satisfiable.
	To this end, $B$ has the states \mbox{$\set{p, p_0, p_f} \cup \Set{p^j_i}{j \in [1..t], i \in [1..\ell-1]}$}, where $p$ is the initial, and $p_f$ is the final state.
	For a fixed $j \in [1..t]$, the states $p^j_1, \dots p^j_{\ell-1}$ are used to iterate through the $\ell$ clauses of $\varphi_j$.
	The transitions between $p^j_i$ and $p^j_{i+1}$ are labeled by $(x_s, ?0)$ or $(x_s, ?1)$, depending on whether variable $x_s$ occurs with or without negation in the $i+1$-st clause of $\varphi_j$ for $i \in [1..\ell-2]$.
	Note that there are at most three transitions between $p^j_i$ and $p^j_{i+1}$.
	For the first clause, the transitions start in $p_0$ and end in $p^j_1$, while for the $\ell$-th clause, the transitions start in $p^j_{\ell-1}$ and end in $p_f$.
	Further, there is a transition from $p$ to $p_0$ that is labeled by $\#$.
	
	To answer the requests of $B$, we use the memory automaton $M$. It ensures that each request of the form $(x_s, ?1)$ is also followed by a confirmation of the form $(x_s, !1)$ (same for value $0$).
	$M$ has an initial state $q_{\mathit{init}}$, a final state $q_{f}$, and for each variable $x_s$, two states $q^0_s$ and $q^1_s$.
	We get a transition between $q_f$ and each $q^0_s$, labeled by the request $(x_s, ?0)$.
	To get the confirmation, we introduce a transition from $q^0_s$ back to $q_f$, labeled by $(x_s, !0)$.
	We proceed similarly for $q_f$ and $q^1_s$.
	To get from $q_{\mathit{init}}$ to $q_f$, we introduce the transition $q_{\mathit{init}} \xrightarrow{\#} q_f$, avoiding that $M$ accepts the empty word. \\
	
	Now we show the correctness of the construction:
	There is a $j\in[1..t]$ such that $\varphi_j$ is satisfiable if and only if $L(S) \cap \Context(\Sigma, k+1, 2\ell) \neq \emptyset$. 
	
	First, let a $j \in [1..t]$ be given such that $\varphi_j$ is satisfiable.
	Then there is a value $v_s$ for each variable $x_s$ with $s \in [1..k]$, satisfying $\varphi_j$.
	A word $w \in L(S, 2\ell)$ can be constructed as follows. 
	Let $x_{s_1}, \dots, x_{s_\ell}$ denote variables (repetition allowed) that contribute to satisfying the clauses of $\varphi_j$.
	This means that $x_{s_i}$ can be used to satisfy the $i$-th clause. 
	We define $w = \#.(x_{s_1}, ? v_{s_1}) . (x_{s_1}, ! v_{s_1}) \dots (x_{s_\ell}, ? v_{s_\ell}) . (x_{s_\ell}, ! v_{s_\ell})$.
	Then, $w \in L(S) \cap \Context(\Sigma, k+1, 2\ell)$.
	
	For the other direction, let $w \in L(S) \cap \Context(\Sigma, k+1, 2\ell)$ be given.
	Then $w$ is of the following form: $w = \# . (x_{s_1}, ? v_{s_1}) . (x_{s_1}, ! v_{s_1}) \dots (x_{s_\ell}, ? v_{s_\ell}) . (x_{s_\ell}, ! v_{i_\ell})$.
	Note that $x_{s_i} = x_{s_{i'}}$ implies $v_{s_i} = v_{s_{i'}}$ in the word.
	Hence, we can construct a satisfying assignment $v$ for one of the given $\varphi_j$.
	We assign each $x_{s_i}$ occurring in $w$ the value $v_{s_i}$.
	For variables that do not occur in $w$, we can assign $0$ or $1$.
	
	By construction, $B$ iterates through the clauses of one of the given $\varphi_j$.
	Since $B$ also accepts in the computation of $w$, there is a $j \in [1..t]$ such that all the clauses of $\varphi_j$ can be satisfied by $v$.
	Hence, $\varphi_j$ is satisfiable. \\
	
	Finally, the parameters of the constructed $\BCS$-instance are the size of the memory, \mbox{$m = 2k + 2$} and the number of context switches $\cs = 2\ell$.
	Both are bounded by $\maximum_{j \in [1..t]} \abs{\varphi_j}$.
	Hence, all requirements on a cross-composition are met.
\end{proof}
\section{Proofs for Section~\ref{Section:RR}}
\label{Appendix:RR}

\subsection{Carving-width}\label{Appendix:CarvingWidth}
The scheduling dimension is closely related to the \emph{carving-width} of an undirected multigraph.
The carving-width was introduced in \cite{Seymour1994} as a measure for communication graphs.
These are graphs where each edge-weight represents a number of communication demands (calls) between two vertices, or locations.
To route these calls efficiently, one is interested in finding a routing tree that minimizes the needed bandwidth.
The \emph{carving width} measures the minimal required bandwidth among all such trees.

To relate it with the scheduling dimension, we turn a directed multigraph $G = (V,E)$ into an undirected multigraph $G' = (V, E')$ the following way:
We keep the vertices $V$ of $G$ and assign the edge-weights $E'(u,v) = \max\set{ E(u,v), E(v,u) }$ for $u,v \in V$.
Then the following holds:
\begin{lemma}\label{Lemma:CWisSdim}
	For any directed multigraph $G$, we have $\sdimof{G} \leq \cwof{G'} \leq 2\sdimof{G}$.
\end{lemma}

Despite the close relation between carving-width and scheduling dimension, we suggest a parameterization in terms of the latter.
The reason is as follows.
The scheduling dimension is the natural measure for directed communication demands in scheduling graphs.
If threads are tightly coupled, they should be grouped together (contracted) to one thread.
This leads to a contraction process rather than to a carving decomposition that is needed for the carving-width.

Before we give the proof of Lemma \ref{Lemma:CWisSdim}, we formally introduce the carving-width.
Let $G = (V,E)$ be a given undirected multigraph.
A \emph{carving decomposition} of $G$ is a tuple $(T, \varphi)$, where $T$ is a binary tree and $\varphi$ is a bijection from the leaves of $T$ to the vertices $V$ of $G$.
For an edge $e$ of $T$, removing $e$ from $T$ partitions $T$ into two connected components.
Let $S_1, S_2 \subseteq V$ be the images, under $\varphi$, of the leaves falling into the components.
We define the \emph{width} of $e$ to be the integer $E(S_1,S_2) = \sum_{u \in S_1, v \in S_2} E(u,v)$, and the \emph{width} of the decomposition $(T, \varphi)$ to be the maximum width of all edges in $T$.
The \emph{carving-width} of $G$ is the minimum width among all carving decompositions:
\begin{align*}
\cwof{G}=\minimum\Set{ \widthof{(T,\varphi)} }{ (T,\varphi) \text{ a carving decomposition of } G }.
\end{align*}
Deciding whether the carving-width of a graph is bounded by a given integer is an $\NP$-hard problem for general graphs and known to be polynomial for planar graphs \cite{Seymour1994}.
The first $\FPT$-algorithm for this decision problem, parameterized by the carving-width, was derived by Bodlaender et al. in \cite{Bodlaender2005}.
A parameterization by the number of vertices of the given graph was considered by Fomin et al. in \cite{Fomin2010}.
They constructed an $\bigOS{2^n}$-time algorithm for computing the carving-width.
Further, the carving-width was used as a parameter in a graph-embedding problem in \cite{Biedl2013}.
The authors used dynamic programming on carving decompositions to show the fixed-parameter tractability of their problem.
\begin{proof}[Proof of Lemma \ref{Lemma:CWisSdim}]
	First we show that from a given carving decomposition $\carvingdecomp$ of $G'$ of width $k$, we can construct a contraction process $\contractionprocess$ of $G$ with degree at most $k$.
	Then we get that $\sdimof{G} \leq \cwof{G'}$.
	The idea is to inductively assign each node of $T$ a partial contraction process of $G$ such that all graphs appearing in the process have degree at most $k$.
	We start at the leaves of $T$ and go bottom-up.
	At the end, the needed contraction process will be the process assigned to the root of $T$.
	
	Before we start, we fix some notation.
	Let $w$ be a vertex occurring in a partial contraction processes starting in $G$.
	By $V(w) \subseteq V$, we denote the set of vertices in $G$ that get contracted to $w$.
	Note that if two vertices $u,v$ get contracted to $w$ during the process, we get that $V(w) = V(u) \cup V(v)$.
	For a node $n$ of $T$, we use $\Leafof{n} \subseteq V$ to denote the image, under $\varphi$, of the leaves of the subtree of $T$ rooted in $n$.
	
	Now we show the following:
	We can assign any node $n$ in $T$ a pair $(\contractionprocess_n, w)$, where $\contractionprocess_n = G_1 . \dots G_\ell$ is a partial contraction process such that $G_1 = G$,
	$\degreeof{G_i} \leq k, i \in [1..\ell]$ and $w$ is a vertex in $V(G_\ell)$ with $V(w) = \Leaf(n)$, and $V(G_\ell) = V \setminus \Leafof{n} \cup \set{w}$.
	The latter conditions ensure that the process contracts the vertices in $\Leafof{n}$ to $w$ and furthermore, no other vertices in $G$ are contracted.
	The process that we assign to the root $r$ is thus a proper contraction process of $G$, contracting all vertices of $G$ to a single node.
	Moreover, the degree of the process is bounded by $k$.
	
	To start the induction, we assign any leaf $n$ of $T$ the pair $(G,\varphi(n))$.
	Note that we have $\Leafof{n} = \set{\varphi(n)} = V(\varphi(n))$ in this case.
	Hence, we only need to elaborate on the degree of $G$ since the remaining conditions above are satisfied. 
	For any $v \in V$, we have:
	\begin{align*}
	\degreeof{v} &= \maximum\set{\indegreeof{v},\outdegreeof{v}}
	= \maximum\set{\sum_{u \in V} E(u,v), \sum_{u \in V} E(v,u)} \\
	&\leq \sum_{u \in V}\maximum\set{E(u,v), E(v,u)}
	= \sum_{u \in V} E'(v,u) = E'(v, V \setminus \set{v}).
	\end{align*}
	Let $n'$ denote the leaf with $\varphi(n') = v$.
	Furthermore, let $e$ be the edge of $T$ connecting $n'$ with its parent node.
	Then $\widthof{e} = E'(v, V \setminus \set{v})$.
	Since the width of $e$ is bounded by $k$, we also get that $\degreeof{v} \leq k$ and thus, $\degreeof{G} \leq k$.
	
	Now suppose we have a node $n$ of $T$ with two children $n_1$ and $n_2$ that are already assigned pairs $(\contractionprocess_1,w_1)$ and $(\contractionprocess_2,w_2)$ with partial contraction processes $\contractionprocess_1 = G_1 \dots G_\ell$ and \mbox{$\contractionprocess_2 = H_1 \dots H_t$}, where $G_1 = H_1 = G$ and $\degreeof{G_i}, \degreeof{H_j} \leq k$ for $i \in [1..\ell], j \in [1..t]$.
	Furthermore, $w_1 \in V(G_\ell)$ and $w_2 \in V(H_{t})$ satisfy the conditions:
	$V(w_1) = \Leafof{n_1}$, \mbox{$V(G_\ell) = V \setminus \Leafof{n_1} \cup \set{w_1}$}, 
	and $V(w_2) = \Leafof{n_2}$, $V(H_t) = V \setminus \Leafof{n_2} \cup \set{w_2}$.
	We also fix the notation for the contractions applied in $\contractionprocess_2$.
	Let $\sigma_i$ be the contraction applied to $H_i$ to obtain $H_{i+1}$.
	Hence, $H_{i+1} = H_i[\sigma_i]$ for $i \in [1..t-1]$.
	
	We construct the partial contraction process that performs the contractions of $\contractionprocess_1$, the contractions of $\contractionprocess_2$, and contracts $w_1$ and $w_2$ to a node $w$. 
	Set $\contractionprocess_n = G_1 \dots G_\ell . G_{\ell + 1} \dots G_{\ell + t}$, where $G_{\ell + i} = G_{\ell + i - 1}[\sigma_i]$, for $i \in [1..t-1]$ and $G_{\ell + t} = G_{\ell + t - 1}\contract{w_1,w_2}{w}$.
	Then $\contractionprocess_n$ is well-defined.
	Since $\Leafof{n_1} \cap \Leafof{n_2} = \emptyset$, we have that \mbox{$V(G_\ell) = V \setminus \Leafof{n_1} \cup \set{w_1}$} contains $\Leafof{n_2}$.
	Thus, it is possible to apply the contractions $\sigma_1, \dots, \sigma_{t-1}$ to $G_\ell$ since they only contract vertices from $\Leafof{n_2}$.
	Assume that a node $v \in V \setminus \Leafof{n_2}$ would be contracted during $\contractionprocess_2$.
	Then $v \notin V(H_t) = V \setminus \Leafof{n_2} \cup \set{w_2}$.
	Since $v \neq w_2$, we would get that $v$ is in $\Leafof{n_2}$ which is a contradiction.
	
	We assign $n$ the pair $(\contractionprocess_n, w)$.
	What is left to prove is that the above conditions are satisfied.
	For the vertex $w \in V(G_{\ell + t})$, we have:
	\begin{align*}
	V(w) = V(w_1) \cup V(w_2) = \Leafof{n_1} \cup \Leafof{n_2} = \Leafof{n}.
	\end{align*}
	Since we apply the contractions of $\contractionprocess_1$ and $\contractionprocess_2$ to obtain $G_{\ell + t - 1}$, we get:
	\begin{align*}
	V(G_{\ell + t - 1}) &= ( V(G_\ell) \cap V(H_t) ) \cup \set{w_1,w_2} \\
	&=  ( V \setminus \Leafof{n_1} \cap V \setminus \Leafof{n_2} ) \cup \set{w_1, w_2} \\
	&= V \setminus ( \Leafof{n_1} \cup \Leafof{n_2} ) \cup \set{w_1, w_2} \\
	&= V \setminus \Leafof{n} \cup \set{w_1, w_2}.
	\end{align*}
	The graph $G_{\ell + t}$ is obtained by contracting $w_1$ and $w_2$ in $G_{\ell + t - 1}$.
	Hence, we have that $V(G_{\ell + t}) = V(G_{\ell + t - 1}) \setminus \set{w_1, w_2} \cup \set{w} = V \setminus \Leafof{n} \cup \set{w}$.
	
	No we prove that all occurring graphs in $\contractionprocess_n$ have degree bounded by $k$.
	It is clear by assumption that this holds for $G_1, \dots, G_\ell$.
	We show the same for $G_{\ell + i}$ with $i \in [1..t-1]$.
	Let $u \in V(G_{\ell + i})$. 
	We distinguish three cases.
	
	If $V(u) \subseteq \Leafof{n_1}$, then we have that none of the $\sigma_j$ act on $u$ since this would imply that a node from $\Leafof{n_1}$ gets contracted by $\sigma_j$ which is not possible.
	Hence, $u \in V(G_{\ell})$ and $\degreeofgraph{G_{\ell + i}}{u} = \degreeofgraph{G_{\ell}}{u} \leq k$.
	Note that by $\degreeofgraph{H}{v}$ we indicate the degree of vertex $v$ in graph $H$.
	
	If $V(u) \subseteq \Leafof{n_2}$, then no contraction of $\contractionprocess_1$ acts on $u$ and $u$ is a vertex that occurs during the application of $\sigma_1, \dots, \sigma_i$.
	Hence, $u \in V(H_{i+1})$ and $\degreeofgraph{G_{\ell + i}}{u} = \degreeofgraph{H_{i+1}}{u} \leq k$.
	
	If $V(u) \subseteq V \setminus (\Leafof{n_1} \cup \Leafof{n_2})$, then $u$ is neither involved in the contractions of $\contractionprocess_1$ nor in the contractions of $\contractionprocess_2$.
	Hence, $u \in V$ and we have: $\degreeofgraph{G_{\ell + i}}{u} = \degreeofgraph{G}{u} \leq k$.
	
	Finally, we prove that the graph $G_{\ell + t}$ has degree bounded by $k$.
	For a vertex $u \neq w$ in $V(G_{\ell + t})$, we have $\degreeofgraph{G_{\ell + t}}{u} = \degreeofgraph{G_{\ell + t - 1}}{u}$ since $u$ is not involved in the contraction $\contract{w_1, w_2}{w}$ that is applied to $G_{\ell + t - 1}$ in order to obtain $G_{\ell + t}$.
	Now we consider $w$.
	First note, that $\indegreeofgraph{G_{\ell + t}}{w} = E(V \setminus V(w), V(w))$ and $\outdegreeofgraph{G_{\ell + t}}{w} = E(V(w), V \setminus V(w))$.
	Then we can derive:
	\begin{align*}
	\degreeofgraph{G_{\ell + t}}{w} 
	&= \maximum \set{\indegreeofgraph{G_{\ell + t}}{w}, \outdegreeofgraph{G_{\ell + t}}{w}} \\
	&= \maximum \set{E(V \setminus V(w), V(w)), E(V(w), V \setminus V(w))} \\
	&= \maximum \set{\sum_{u \in V(w), v \in V \setminus V(w)} E(v,u), \sum_{u \in V(w), v \in V \setminus V(w)} E(u,v) } \\
	&\leq \sum_{u \in V(w), v \in V \setminus V(w)} \maximum \set{E(v,u), E(u,v)} \\
	&= \sum_{u \in V(w), v \in V \setminus V(w)} E'(v,u) \\
	&= E'(V(w), V \setminus V(w)).
	\end{align*}
	Let $e$ denote the edge between $n$ and its parent node.
	Then $\widthof{e} = E'(V(w), V \setminus V(w))$.
	Since the width is bounded by $k$, we get that also $\degreeofgraph{G_{\ell + t}}{w}$ is bounded by $k$ and hence, $\degreeof{G_{\ell + t}} \leq k$.
	Note that in the case where $n$ is the root, the degree of $w$ is $0$. \\
	
	To prove that $\cwof{G'} \leq 2\sdimof{G}$, we show how to turn a given contraction process $\contractionprocess$ of $G$ with degree $k$ into a carving decomposition $\carvingdecomp$ of $G'$ with width at most $2k$.
	
	Let $\contractionprocess = G_1, \dots, G_{\abs{V}}$ be the given process.
	We inductively construct a tree $T$ with a labeling $\nodeLab : V(T) \rightarrow \bigcup_{i \in [1..\abs{V}]} V(G_i)$ that assigns to each node in $T$ a vertex from one of the $G_i$.
	We start with a root node $r$ and set $\nodeLab(r) = w$, where $w$ is the latest vertex that was introduced by the contraction process: $G_{\abs{V}} = G_{\abs{V} - 1}\contract{w_1, w_2}{w}$.
	
	Now suppose, we are given a node $n$ of $T$ with $\nodeLab(n) = v$.
	Assume $v$ occurs in $\contractionprocess$ on the right hand side of a contraction.
	This means there is a $j$ such that $G_{j+1} = G_j \contract{v_1, v_2}{v}$.
	We add two children $n_1$ and $n_2$ to $T$ and set $\nodeLab(n_i) = v_i$ for $i = 1, 2$.
	If $v$ does not occur on the right hand side of a contraction in $\contractionprocess$ then $v$ is a vertex of $G$.
	In this case, we stop the process on this branch and $n$ is a leaf of $T$.
	
	Hence, we obtain a tree $T$ where the leaves are labeled by vertices from $G$.
	If we set $\varphi$ to be $\nodeLab$ restricted to the leaves, then $\varphi$ is a bijection between the leaves of $T$ and $V$ and $\carvingdecomp$ is a carving decomposition of $G'$.
	
	Now we show by induction on the structure of $T$ that for each node $n$ of $T$ we have: $V(\nodeLab(n)) = \Leafof{n}$.
	Recall that $\Leafof{n}$ is the image, under $\varphi$, of the leaves of the subtree of $T$ rooted in $n$.
	We start at the leaves of $T$.
	Let $l$ be a leaf, then we have: $\Leafof{l} = \set{\lambda(l)}$ and moreover $V(\nodeLab(l)) = \set{\nodeLab(l)}$.
	For a node $n$ of $T$ with children $n_1, n_2$ such that the equations $\Leafof{n_i} = V(\nodeLab(n_i))$ already hold for $i = 1, 2$, we get:
	\begin{align*}
		\Leafof{n} = \Leafof{n_1} \cup \Leafof{n_2} = V(\nodeLab(n_1)) \cup V(\nodeLab(n_2)).
	\end{align*}
	Since $n_1, n_2$ are the children of $n$, we get by the construction of $T$ that there is a contraction in $\contractionprocess$ of the form $G_{j+1} = G_j\contract{\nodeLab(n_1), \nodeLab(n_2)}{\nodeLab(n)}$ and hence:
	\begin{align*}
		V(\nodeLab(n_1)) \cup V(\nodeLab(n_2)) = V(\nodeLab(n)).
	\end{align*}
	Finally, we show that the width of the carving decomposition $\carvingdecomp$ is at most $2k$.
	To this end, let $e$ be an edge in $T$, connecting the node $n$ with its parent node.
	Further, let $w$ denote $\nodeLab(n)$ and $w \in V(G_j)$.
	Then we have:
	\begin{align*}
		\widthof{e} &= E'(\Leafof{n}, V \setminus \Leafof{n}) \\
		&= E'(V(w), V \setminus V(w)) \\
		&= \sum_{u \in V(w), v \in V \setminus V(w)} E'(u,v) \\
		&= \sum_{u \in V(w), v \in V \setminus V(w)} \maximum \set{E(u,v), E(v,u)} \\
		&\leq \sum_{u \in V(w), v \in V \setminus V(w)} (E(u,v) + E(v,u)) \\
		&= E(V(w), V \setminus V(w)) + E(V \setminus V(w), V(w)) \\
		&= \outdegreeofgraph{G_j}{w} + \indegreeofgraph{G_j}{w}.
	\end{align*}
	Since $\degreeof{\contractionprocess}$ is bounded by $k$, also the degree of $G_j$ is bounded by $k$ and hence, $\widthof{e}$ is at most $2k$.
	All in all, the width of $\carvingdecomp$ is at most $2k$.
\end{proof}
\subsection{Correctness and Complexity of $\bcslsd$}\label{Appendix:BCSL}
We first show the correctness of the stated fixed-point iteration by proving Lemma \ref{Lemma:lfp}.
\begin{proof}[Proof of Lemma \ref{Lemma:lfp}]
	First, suppose that $L(S)\cap \sdimlang \neq \emptyset$.
	Then there exists a word $u$ in \mbox{$L(S)\cap \sdimlang$} with scheduling graph $G(u) = (V,E)$.
	We may assume that $V = [1..t]$.
	This means that all given threads participate in the computation.
	If this is not the case, we can delete the non-participating threads in the instance.
	By assumption, we know that $\sdimof{\csgraphof{u}}\leq \sdim$.
	Hence, there is a contraction process $\contractionprocess=\mgraph_1, \ldots, \mgraph_{\abs{V}}$ of $\csgraphof{u}$ such that $\degreeof{\contractionprocess} \leq \sdim$.
	
	We now associate to each node in $\mgraph_i$, an element from $(Q \times Q)^{\leq \sdim} \times \Powerset([1..t])$.
	To this end, let $u = u_1 \dots u_m$ be the unique context decomposition of $u$ with respect to to the run of $S$ on $u$. 
	Furthermore, let $q_j$ be the (memory) state of $M$, reached after reading $u_1 \dots u_j$ with $j \in [1..m]$.
	Note that $q_m = q_{\mathit{final}}$ and we set $q_0 = q_{\mathit{init}}$.
	Then we get the interface sequence $\alpha = (q_0,q_1)(q_1,q_2) \dots (q_{m-1},q_m)$ by taking the pair of states corresponding to each context $u_j$.
	
	From $\alpha$, we obtain the interface sequence $\sigma_i$, for each thread $i \in [1..t]$, by deleting from $\alpha$ the pairs of states of the contexts in which thread $i$ was not active.
	Note that the length (the number of pairs) of $\sigma_i$ is the number of times process $i$ is active.
	Further, this is the degree of $i$ in $\csgraphof{u}$.
	
	Now we use the obtained interface sequences to tag the nodes in $\mgraph_1$.
	We define the map $\ifmap{1}:[1..t] \rightarrow (Q\times Q)^{\leq \sdim} \times \Powerset([1..t])$ by $\ifmap{1}(i) = (\sigma_i,\set{i})$ for $i \in [1..t]$.
	Clearly $\ifmap{1}(i) \in L_1$ for any $i$.
	
	Given a map $\ifmap{j} : V(\mgraph_j) \rightarrow (Q\times Q)^{\leq \sdim} \times \Powerset([1..t])$ with $j < t$, we inductively construct a map $\ifmap{j+1}$ from the nodes of $\mgraph_{j+1}$ to $ (Q\times Q)^{\leq \sdim}\times\Powerset([1..t]) $.
	Let $\mgraph_{j+1} = \mgraph_j \contract{n_1, n_2}{n}$.
	Then $V(\mgraph_{j+1}) = (V(\mgraph_j) \setminus \set{n_1,n_2}) \cup \set{n}$.
	For $v \in V(\mgraph_j) \setminus \set{n_1,n_2}$, we set $\ifmap{j+1}(v) = \ifmap{j}(v)$.
	For the image of $n$, let $\ifmap{j}(n_1) = (\tau_1, T_1)$ and $\ifmap{j}(n_2) = (\tau_2, T_2)$. 
	Let $T$ denote the union $T_1 \cup T_2$.
	Further, let $\sigma_n$ be obtained from $\alpha$ as follows:
	First mark all the pairs of states in $\alpha$ that correspond to a thread $i$ in $T$.
	We concatenate any two adjacent pairs that are marked.
	If $(q_{i-1}, q_i)(q_i, q_{i+1})$ are marked, then we concatenate it to $(q_{i-1},q_{i+1})$ and mark the resultant pair. 
	We do this until we can no longer find an adjacent marked pair.
	Now we delete all the memory pairs that remain unmarked.
	We denote the resulting interface sequence by $\sigma_n$ and define: $\ifmap{j+1}(n) = (\sigma_n, T)$. 
	
	Note that concatenating adjacent marked pairs corresponds to deleting edges between $T_1$ and $T_2$ in $\csgraphof{u}$. 
	Hence, it is the same as contracting the corresponding nodes $n_1$ and $n_2$ in the graph $\mgraph_j$.
	We get that the length of $\sigma_n$ is the degree of $n$ in $G_j$, which is bounded by $\sdim$.
	Thus, $\ifmap{j+1}(n)$ is an element in $(Q\times Q)^{\leq \sdim}$ and in $\ifmap{j}(n_1) \mergeop^k \ifmap{j}(n_2) \subseteq L_{j+1}$. 
	
	The map $\ifmap{t}$ is a map from a single element $V(G_t) = \set{z}$ to $(Q\times Q)^{\leq \sdim} \times \Powerset([1..t])$.
	We get that $\ifmap{t}(z) = ((q_{0}, q_{m}),[1..t]) = ((q_{\mathit{init}}, q_{\mathit{final}}),[1..t]) \in L_{t+1} = L_t$. \\
	
	For the other direction, assume that $((q_{\mathit{init}}, q_{\mathit{final}}), T)  \in L_{m}$ for an $m \in \Nat$ and $T \subseteq [1..t]$. 
	We may assume that $T = [1..t]$.
	Otherwise, we delete the non-participating threads from the given instance.
	We show that  $L(S)\cap \sdimlang \neq \emptyset$. 
	To this end, we first construct an execution tree $\execTree$ together with a labeling $\nodeLab : V(\execTree) \rightarrow (Q \times Q)^{\leq \sdim} \times \Powerset([1..t])$, based on the interface sequences that were used to obtain $(q_{\mathit{init}}, q_{\mathit{final}})$.
	
	We start with a single root node $r$ and set $\nodeLab(r) = ((q_{\mathit{init}}, q_{\mathit{final}}), [1..t])$.
	Now given a partially constructed execution tree, we show how to extend it.
	If for all leaves $l$ of the constructed tree we have $\nodeLab(l) = (\tau, T)$, where $\abs{T} = 1$ then we stop.
	Otherwise, we pick a leaf $l$ with $\abs{T} > 1$. 
	Then there are generalized interface sequences $(\tau_1, T_1)$ and $(\tau_2, T_2)$ such that $(\tau, T) \in (\tau_1, T_1) \mergeop^k (\tau_2, T_2)$.
	Note that $(\tau_1, T_1)$ and $(\tau_2, T_2)$ are not unique.
	But we can arbitrarily pick any pair of them.
	To extend the tree, we add two nodes $l_1$ and $l_2$ and set $\nodeLab(l_i) = (\tau_i, T_i)$ for $i = 1, 2$.
	
	The procedure clearly terminates and yields an execution tree $\execTree$ where the leaves $l$ satisfy: 
	$\nodeLab(l) \in L_1$.
	Hence, the leaves show the interface sequences that were used to obtain $((q_{\mathit{init}}, q_{\mathit{final}}), [1..t])$ by the fixed point algorithm.
	
	Now we make use of the tree to construct a word in $L(S)$ with scheduling graph of bounded scheduling dimension.
	To obtain the word, we need to inductively define the map $\Pi: V(\execTree) \rightarrow (Q \times Q)^*$.
	We start at the leaves.
	For a leaf $l$, we set $\Pi(l) = \tau$, where $\tau$ is the first component of $\nodeLab(l)$: $\nodeLab(l) = (\tau, \set{i})$.
	Note that $\tau \in \ilangof{A_i}$.
	This means that for $\tau = (q_{i_1},q'_{i_2}) (q_{i_2},q'_{i_3}) \dots (q_{i_m},q'_{i_{m+1}})$ there are words $u^i_1, \dots, u^i_m$ such that $u^i_1 \dots u^i_m \in L(A_i)$ and $u^i_j \in L(M(q_{i_j},q'_{i_{j+1}}))$, for $j \in [1..m]$.
	
	Let $l$ be a node in $\execTree$ with children $l_1$ and $l_2$.
	Further, let $\nodeLab(l) = (\tau, T)$, $\Pi(l_1) = \tau_1'$ and $\Pi(l_2) = \tau_2'$.
	We set $\Pi(l) = \tau'$, where $\tau' \in \tau_1' \SShuffle \tau_2'$ and $\tau \in \tau' \! \closure$.
	As before, $\tau'$ does not need to be unique.
	But we can pick any of them, satisfying the requirements.
	We stop the procedure if we assigned the root a value under $\Pi$.
	
	Now we have that for any node $l$ in $\execTree$ with $\Pi(l) = (q_{i_1},q'_{i_2}) (q_{i_2},q'_{i_3}) \dots (q_{i_m},q'_{i_{m+1}})$ and $\nodeLab(l) = (\tau, T)$, there are words $u_1, \dots, u_m$ such that $u_1 \dots u_m \in \Shuffle{i \in T}{} L(A_i)$.
	For the root $r$ this means that there is a word $u$ which lies in $\Shuffle{i \in [1..t]}{} L(A_i)$ and in $L(M)$.
	Hence, $u \in L(S)$.

	It is left to show that $u$ has a scheduling graph of bounded scheduling dimension.
	To this end, consider the interface sequence associated to $r$: $\Pi(r) = (q_0,q_1)(q_1,q_2) \dots (q_{m-1},q_m)$.
	For each tuple \mbox{$(q_j, q_{j+1}), j \in [1..m-1]$}, there is a unique leaf $l_j$ in $\execTree$ such that $(q_j, q_{j+1})$ belongs to the interface sequence \mbox{$\Pi(l_j) = \tau_j$}.
	Let $\nodeLab(l_j) = (\tau_j, \set{c_j})$.
	Then we fix the order in which the thread take turns to: $c_0, \dots, c_{m-1}$.
	Note that $c_0$ is the thread corresponding to $(q_0, q_1)$, $c_1$ is the thread corresponding to $(q_1,q_2)$ and so on.
	Clearly, the computation of $S$ reading the word $u$ follows the described order.
	It is thus easy to construct the scheduling graph $G = \csgraphof{u}$.
	
	In order to show that $G = (V, E)$ has scheduling dimension bounded by $\sdim$, we first consider the undirected multigraph $G' = (V, E')$.
	Recall that we obtain $G'$ by taking all the vertices of $G$ and setting $E'(u,v) = \max \set{ E(u,v), E(v,u) }$ for $u,v \in V$.
	Now for any leaf $l_j$ of $\execTree$, we set $\varphi(l_j) = c_j$, where $\nodeLab(l_j) = (\tau_j, \set{c_j})$. 
	Then $(\execTree, \varphi)$ is a carving decomposition of $G'$.
	We show that the decomposition has width at most $\sdim$.
	Consider any edge $(n,k)$ of $\execTree$, where $n$ is a child node of $k$.
	Let $\nodeLab(n) = (\tau, T)$.
	Then removing the edge $(n,k)$ from $\execTree$ partitions the vertices of $G'$ into $T$ and $V \setminus T$.
	Now note that the number of pairs in $\tau$ shows how often $T$, seen as one thread, participates in the computation of $S$ on $u$.
	Hence, we get $E'(T, V \setminus T) \leq \abs{\tau}$.
	As $\abs{\tau}$ is bounded by $\sdim$, we get that the width of $(n,k)$ is also bounded by $\sdim$.
	Hence, $\widthof{(T, \varphi)} \leq \sdim$.
	Finally, by Lemma \ref{Lemma:CWisSdim}, we get that $\sdim(G) \leq \sdim$.
\end{proof}
It remains to estimate the complexity of computing the fixed point.
Since the generalized product requires disjoint sets of threads, the computation will stop after $t$ steps.
Each step has to go over at most $(m^{2(\sdim+1)} 2^t)^2 = m^{4\sdim + 4}4^{t}$ combinations of generalized interface sequences.
Computing each such composition $(\sigma_1, S_1)\mergeop^{\sdim} (\sigma_2, S_2)$ requires us to consider all $\rho\in \sigma_1 \SShuffle \sigma_2$.
Forming a shuffle of $\sigma_1$ and $\sigma_2$ can be understood as setting $\abs{\sigma_2}$ bits in a bitstring of length $\abs{\sigma_1}+\abs{\sigma_2}$.  
Hence, the number of shuffles $\rho$ is \mbox{${{\abs{\sigma_1}+\abs{\sigma_2}}\choose{\abs{\sigma_2}}}\leq 2^{2\sdim}=4^{\sdim}$}.
Given $\rho$, we determine $\rho\closure$ by iteratively forming summaries.
In the worst case, $\rho$ has length $2\sdim$.
We mark an even number of positions in $\rho$ and summarize the intervals between every pair of markers $2i$ and $2i+1$.
Since there are at most $\sdim$ even positions, we obtain
$\sum_{i=0}^{\sdim}{{2\sdim}\choose{2i}}\leq 4^{\sdim}$ elements in $\rho\closure$.
All in all, the effort is \mbox{$t m^{4\sdim + 4} 4^t 16^{\sdim} =   \bigOS{(2m)^{4\sdim}4^{t}}$}.
\subsection{Correctness and Complexity of $\bcslfix$}\label{Appendix:BCSLFIX}
Before we explain the complexity of the iteration, we prove Lemma \ref{Lemma:fixedlfp}.
\begin{proof}[Proof of Lemma \ref{Lemma:fixedlfp}]
	First, suppose that $\bcslfix$ holds on the instance $(S, G, \contractionprocess)$.
	This means that there is a word $u$ in $L(S)$ such that $\csgraphof{u} = G = (V,E)$.
	We may assume that $V = [1..t]$.
	Further, let $\contractionprocess = G_1, \dots, G_{t}$ be the contraction process.
	
	We proceed as in the proof of Lemma \ref{Lemma:lfp} and construct the maps $\nodeLab_j$ from $V(G_j)$ to $(Q \times Q)^{\leq \sdim} \times \Powerset([1..t])$.
	This time we get that $\nodeLab_1(v) \in S_v \times \Powerset([1..t])$ for all $v \in V$. 
	Moreover, for each contraction $G_{j+1} = G_j\contract{n_1, n_2}{n}$, we get: 
	$\nodeLab_{j+1}(n) \in (S_{n_1} \omerge{i}{k} S_{n_2})$, where $i = E(n_1,n_2)$ and $k = E(n_2,n_1)$.
	Note that these are the edge weights in $G_j$.
	Hence, $\nodeLab_{j+1}(n) \in S_n \times \Powerset([1..t])$.
	Then we also get that $\nodeLab_{t}(w) = ((q_{\mathit{init}} , q_{\mathit{final}}), [1..t]) \in S_w \times \Powerset([1..t])$. \\
	
	For the other direction, we show that $(q_{\mathit{init}}, q_{\mathit{final}}) \in S_{w}$ implies the existence of a word $u \in L(S)$ such that the scheduling graph of $u$ is the given graph $G$.
	Again, we may assume that $V = [1..t]$.
	Our goal is to construct an execution tree $\execTree$ together with a labeling $\nodeLab$ as in Lemma \ref{Lemma:lfp}.
	
	We know that the algorithm for $\bcslfix$ computes sets.
	It starts with the initial sets $S_{v}$ for $v \in V$ in the first step and computes further sets along $\contractionprocess$.
	For each contraction $\contract{n_1,n_2}{n}$ the set $S_n$ is given by $S_{n_1} \omerge{i}{k} S_{n_2}$, where $i = E(n_1,n_2)$, $k = E(n_2,n_1)$.
	
	We start the construction of $\execTree$ by setting $\nodeLab(r) = ((q_{\mathit{init}}, q_{\mathit{final}}), V(w))$.
	Recall that $w$ is the only remaining node in $V(G_t)$ and $V(w) = V$ is the set of vertices that contract to $w$ in $\contractionprocess$.
	Moreover, $(q_{\mathit{init}}, q_{\mathit{final}}) \in S_{w}$ by assumption.
	
	Now assume we have a node $l$ in the yet constructed tree such that $\nodeLab(l) = (\tau, T)$, $T = V(n)$ for a vertex $n$ in the contraction process, and $\tau \in S_n$.
	Assume that $\abs{T} > 1$.
	Then there is a contraction $\contract{n_1,n_2}{n}$ in $\contractionprocess$ and interface sequences $\tau_1 \in S_{n_1}$, $\tau_2 \in S_{n_2}$ such that $\tau \in \tau_1 \omerge{i}{k} \tau_2$ with $i = E(n_1,n_2)$, $k = E(n_2,n_1)$.
	We add two nodes $l_1, l_2$ to the tree and set $\nodeLab(l_i) = (\tau_i, V(n_i))$ for $i = 1,2$.
	We stop the process if every constructed node has a labeling $(\tau, T)$ with $\abs{T} = 1$.
	Note that each leaf $l$ in $\execTree$ corresponds to a vertex $v \in V$ and $\nodeLab(l) = (\tau, \set{v})$.

	Now we show how to construct a map $\Pi$ from the nodes of $\execTree$ to the set of \emph{locked} interface sequences as in Lemma \ref{Lemma:lfp}. 
	A \emph{locked} interface sequence is an interface sequence, where adjacent pairs of memory states can be locked.
	This means that, when forming the shuffle with another interface sequence, the locked positions cannot be divided: No other context is allowed to occur between locked pairs.
	
	We start with the leaves of the tree.
	Let $l$ be a leaf with $\nodeLab(l) = (\tau, \set{v})$. 
	Then we set $\Pi(l) = \tau$ without locking any pairs of states.
	
	Now, let $v$ be a node in $\execTree$ with children $v_1$ and $v_2$ such that $\Pi(v_1) = \tau_1$ and $\Pi(v_2) = \tau_2$ are already constructed.
	Let $\nodeLab(v) = (\sigma, V(n))$, $\nodeLab(v_1) = (\sigma_1, V(n_1))$, and $\nodeLab(v_2) = (\sigma_2, V(n_2))$, where $n,n_1$, and $n_2$ are nodes occurring in a contraction $\contract{n_1,n_2}{n}$ of $\contractionprocess$.
	We set $\Pi(v) = \tau$, where $\tau$ is a locked interleaving sequence such that:
	(1) $\tau \in \tau_1 \SShuffle \tau_2$,
	(2) any adjacent pairs that were locked in $\tau_1$ and $\tau_2$ are still adjacent and locked in $\tau$, and
	(3) we find exactly $i$ out-contractions and $k$ in-contractions in $\tau$, where the pairs of states are not locked, and lock them.
	Here, $i = E(n_1,n_2)$ and $j = E(n_2,n_1)$.
	
	We stop the process, when we assigned a value under $\Pi$ to the root $r$.
	Then, in the locked interface sequence $\Pi(r)$ every adjacent pairs of states are locked.
	As in Lemma \ref{Lemma:lfp}, we get that there is a word $u \in L(S)$ following the interface sequence $\Pi(r)$.
	Furthermore, we an construct the sequence $\mathit{ord}$ describing the order in which the threads take turns on $u$.
	From this we get the graph $G(u)$.
	
	In the sequence $\mathit{ord}$, for each two processes $v$ and $v'$, we have that $v'$ appears immediately after $v$ exactly $E(v,v')$ many times.
	From this, we actually get that $G(u) = G$.
\end{proof}
For the complexity, note that the iteration stops after $t$ steps.
Each step has to form at most $(m^{2(\sdim+1)})^2 = m^{4\sdim + 4}$ directed products of interface sequences.
Computing $\sigma \omerge{i}{k} \tau$ can be done similarly to the more general product $\mergeoppar{\sdim}$.
We seek through all $4^{\sdim}$ elements in $\sigma \SShuffle \tau$ and choose the $i+j$ positions where we need to contract.
Hence, the directed products can be computed in time $\bigOS{16^{\sdim}}$, which completes the complexity estimation stated in Theorem \ref{Theorem:bcslfix}.
\subsection{Lower Bound for Round Robin}
\begin{proof}[Proof of Lemma \ref{Theorem:RRLB}]
	We elaborate on the reduction from $\kkClique$ to $\bcslrr$.
	Our goal is to map an instance $(G,k)$ of $\kkClique$ to an instance $(S = \SMS{\Sigma}{M}{A}{i}{t}, \cs)$ of $\bcslrr$ such that $\cs = k$ and $m \leq 2 \cdot k^3$.
	Then a $2^{o(\cs \log(m))}$-time algorithm for $\bcslrr$ would yield an algorithm with runtime
	\begin{align*}
		2^{o( k \log(2 \cdot k^3) )} = 2^{o( 3 k \log(k) + k\log(2) )} = 2^{o( k \log(k) )}
	\end{align*}
	for $\kkClique$.
	This contradicts $\ETH$.
	
	We proceed in two phases:
	A guess-phase where we guess a vertex from each row.
	And a verification-phase where we verify that the guessed vertices induce a clique on $G$ by enumerating all the needed edges among the vertices.
	
	Assume that $V(G) = \Set{v_{ij}}{i,j \in [1..k]}$.
	Vertex $v_{ij}$ is the $j$-th node in row $i$.
	Set $\Sigma = \Set{(v,i), (\#,i)}{v \in V(G), i \in [1..k]}$
	We construct a process $A_i, i \in [1..k]$ for each row of $G$.
	The automaton $A_i$ has $k+1$ states, $q^i_0, \dots, q^i_k$, and the following transitions.
	\begin{itemize}
		\item To pick a vertex from row $i$: 
		$q^i_0 \xrightarrow{(v_{ij}, i)} q^i_j$ for $j \in [1..k]$.
		\item To enumerate edges containing the chosen node: 
		For each $i' < i$ and $j' \in [1..k]$ such that $v_{ij}$ and $v_{i'j'}$ share an edge in $G$, we get: $q^i_j \xrightarrow{(v_{i'j'}, i)} q^i_j$.
		\item For the trivial context: $q^i_j \xrightarrow{(\#, i)} q^i_j$.
	\end{itemize}
	
	We construct the memory automaton $M$ along the two aforementioned phases.
	In the first phase, $M$ runs through each of the $A_i$ and synchronizes on one of the letters $(v_{ij}, i)$, which amounts to picking a vertex from row $i$.
	To this end, $M$ has exactly $k+1$ states $\set{q_1,\dots q_{k+1} }$ and for all $i \in [1..k]$, the transitions: $q_i \xrightarrow{(v_{ij}, i)} q_{i+1}$, where $j \in [1..k]$.
	Thus, in the first phase $M$ reads a word of the form $(v_{1j_1}, 1) \dots (v_{kj_k}, k)$.
	The contribution of $A_i$ to the word is simply $(v_{ij_i}, i)$.
	
	In the second phase, $M$ performs exactly $k-1$ rounds. 
	In the $i$-th round, it first performs the trivial contexts and synchronizes on $(\#, i')$ with $A_{i'}, i' \in [1..i-1]$.
	Then $M$ stores $v_{ij}$ in its states and synchronizes on $(v_{ij}, i')$ with $A_{i'}, i' \in [i+1..k]$.
	Note that $A_{i'}$ can only synchronize with $M$ if the vertex chosen in row $i'$ and $v_{ij}$ share an edge.
	Furthermore, the synchronization in that step is in ascending order: $A_1, \dots, A_k$.
	Hence, as in the first phase, the schedule is round-robin.
	Formally, for round $i$, we have the states $\set{(p^i_1, \bot), \dots, (p^i_i, \bot)}$ for the trivial contexts and $\Set{(p^i_{i'}, j)}{ i' \in [i+1..k], j \in [1..k] }$ for enumerating the edges.
	We also need the last state in round $i$: $(p^i_{k+1}, \bot)$.
	Further, we set $q_{k+1} = (p^1_1, \bot)$ and $(p^{i-1}_{k+1},\bot) = (p^i_1,\bot)$ for $i \in [2..k-1]$ to connect the different rounds.
	The final state of $M$ is the last state in round $k-1$: $(p^{k-1}_{k+1}, \bot)$.
	In round $i$, we get the following transitions:
	\begin{itemize}
		\item To perform the trivial contexts, we get for $i' \in [1..i-1]$: $(p^i_{i'}, \bot) \xrightarrow{(\#, i')} (p^i_{i'+1})$.
		\item To remember the vertex chosen in row $i$, we get for all $j \in [1..k]$: $(p^i_i, \bot) \xrightarrow{(v_{ij}, i)} (p^i_{i+1}, j)$.
		\item For the actual enumeration of the edges, we have for each $i' \in [i+1..k]$ and $j \in [1..k-1]$ the transition $(p^i_{i'}, j) \xrightarrow{(v_{ij}, i')} (p_{i'+1}, j)$.
		\item For the last transition in round $i$, we get for each $j \in [1..k]$: $(p^i_k, j) \xrightarrow{(v_{ij}, k)} (p^i_{k+1}, \bot)$.
	\end{itemize}
	
	Now note that a word of the form $(\#, 1) \dots (\#, i-1) . (v_{ij}, i) \dots (v_{ij}, k)$ is accepted in round $i$ if and only if $v_{ij}$ is the chosen vertex from row $i$ and there is an edge to each of the vertices chosen from the rows $i+1, \dots, k$.
	Hence, vertices $v_{1j_1}, \dots, v_{kj_k}$ form a clique as desired if and only if the word $w = \textbf{Init}. \textbf{Ver}_1 \dots \textbf{Ver}_{k-1} \in L(S)$, where
	\begin{itemize}
		\item $\textbf{Init}$ = $(v_{1j_1}, 1) \dots (v_{kj_k}, k)$, and
		\item $\textbf{Ver}_i = (\#, 1) \dots (\#, i-1) . (v_{ij_i}, i) \dots (v_{ij_i},k)$ for $i \in [1..k-1]$.
	\end{itemize}
	
	Further, the words in $L(S)$ can only be obtained from $k$ rounds of the round-robin schedule and $M$ has at most $2k^3$ many states.
\end{proof}
\section{Proofs for Section~\ref{Section:Discussion}}
\label{Appendix:Discussion}

In Section \ref{Section:Discussion}, we consider the reachability problem in shared memory systems.
We show a hardness result for the parameterization of the problem by the number of threads, as well as an $\FPT$-result if we additionally parameterize by the size of the threads.
The presented $\FPT$-algorithm is optimal.
We show a lower bound based on $\kkClique$.
Formally, the problem is defined as follows:
\begin{problem}
	\problemtitle{Context Switching}
	\problemshort{($\CS$)}
	\probleminput{ An SMCP $S = \SMS{\Sigma}{M}{A}{i}{t}$.}
	\problemquestion{Is $L(S) \neq \emptyset$ ?}
\end{problem}
\subparagraph*{A Hardness Result.}
We show that parameterizing by the number of threads yields the problem $\CS(t)$, which is hard for any level of the $\W$-hierarchy. 
\begin{lemma}
	$\CS(t)$ is $\W[i]$-hard for any $i \geq 1$.
\end{lemma}
For the proof, we reduce from $\BDFAI(n,\abs{\Sigma})$, which is known to be hard for $\W[i], i\geq 1$ \cite{Wareham2000}.
Note that such a reduction is constructed in Lemma \ref{Lemma:BDFAI}.
In fact, the reduction does not change the number of threads and preserves the parameter.
\subparagraph*{Upper Bound.}
Now we add a further parameter, the maximal size of the threads: $a$.
\begin{lemma}
	$\CS(a,t)$ can be solved in time $\bigOS{a^t}$.
\end{lemma}
\begin{proof}
The idea is to run the threads $A_i$ and the memory automaton $M$ concurrently on a product automaton.
The set of states of the product is the set $Q_{A_1} \times \dots \times Q_{A_t} \times Q_M $.
The transition relation is obtained as follows:
From any state $(p_1, \dots, p_i, \dots, p_t, q)$ of the product, we get a transition to $(p_1, \dots, p'_i, \dots, p_t, q')$, labeled by $a \in \Sigma$ if \mbox{$a \in \sync{q}{q'}{i}{p_i}{p'_i}$}.
This means that $A_i$ and $M$ synchronize on the letter $a$.
Note that the language of the product is non-empty if and only if $L(S) \neq \emptyset$.

The product can be build and checked for non-emptiness in $\bigOS{a^t}$ time.
\end{proof}
\subparagraph*{Lower Bound.}
We show the optimality of the above algorithm.
To this end, we give a reduction from $\kkClique$.
\begin{lemma}
 Assuming $\ETH$, $\CS$ cannot be solved in $2^{o(t\log(P))}$ time.
\end{lemma}
\begin{proof}
	The reduction from Lemma \ref{Theorem:RRLB} also applies here.
	Note that the we construct $k$ threads with $k+1$ many states each.
	Furthermore, the memory $M$ enforces a round-robin schedule which can be simulated by $\CS$.
\end{proof}

\end{document}